\documentclass{amsart}

\usepackage{savesym}
\usepackage{txfonts}
\savesymbol{lrcorner}\savesymbol{ulcorner}
\usepackage{amssymb, amscd,stmaryrd,setspace,mathtools,makecell}
\usepackage{mathabx}
\usepackage{enumerate}
\usepackage[usenames,dvipsnames]{xcolor}
\usepackage[bookmarks=true,colorlinks=true, linkcolor=MidnightBlue, citecolor=cyan]{hyperref}
\usepackage{lmodern}
\usepackage{graphicx,float}
\restoresymbol{txfonts}{lrcorner}\restoresymbol{txfonts}{ulcorner}

\input xy
\xyoption{all} \xyoption{poly} \xyoption{knot}\xyoption{curve}
\usepackage{xypic,color}
\newcommand{\comment}[1]{}

\def\tn{\textnormal}

\def\mc{\mathcal}

\def\RR{{\mathbb R}}

\def\PP{{\mathbb P}}
\def\NN{{\mathbb N}}

\def\Hom{\tn{Hom}}
\def\iHom{\Rightarrow}
\def\Aut{\tn{Aut}}

\def\Ob{\tn{Ob}}

\def\SEL*{\tn{SEL*}}

\def\hsp{\hspace{.3in}}

\def\singleton{{\{*\}}}

\def\LoopSchema{{\parbox{.5in}{\fbox{\xymatrix{\LMO{s}\ar@(l,u)[]^f}}}}}

\def\to{\rightarrow}
\def\from{\leftarrow}

\def\taking{\colon}

\def\inj{\hookrightarrow}
\def\surj{\twoheadrightarrow}
\def\too{\longrightarrow}
\def\fromm{\longleftarrow}

\def\ss{\subseteq}

\def\iso{\cong}

\def\|{{\;|\;}}
\def\m1{{-1}}

\def\ol{\overline}
\def\ul{\underline}

\newcommand{\LMO}[1]{\stackrel{#1}{\bullet}}

\def\ullimit{\ar@{}[rd]|(.3)*+{\lrcorner}}
\def\urlimit{\ar@{}[ld]|(.3)*+{\llcorner}}
\def\lllimit{\ar@{}[ru]|(.3)*+{\urcorner}}
\def\lrlimit{\ar@{}[lu]|(.25)*+{\ulcorner}}
\def\ulhlimit{\ar@{}[rd]|(.3)*+{\diamond}}
\def\urhlimit{\ar@{}[ld]|(.3)*+{\diamond}}
\def\llhlimit{\ar@{}[ru]|(.3)*+{\diamond}}
\def\lrhlimit{\ar@{}[lu]|(.3)*+{\diamond}}
\newcommand{\clabel}[1]{\ar@{}[rd]|(.5)*+{#1}}
\newcommand{\TriRight}[7]{\xymatrix{#1\ar[dr]_{#2}\ar[rr]^{#3}&&#4\ar[dl]^{#5}\\&#6\ar@{}[u] |{\Longrightarrow}\ar@{}[u]|>>>>{#7}}}
\newcommand{\TriLeft}[7]{\xymatrix{#1\ar[dr]_{#2}\ar[rr]^{#3}&&#4\ar[dl]^{#5}\\&#6\ar@{}[u] |{\Longleftarrow}\ar@{}[u]|>>>>{#7}}}
\newcommand{\TriIso}[7]{\xymatrix{#1\ar[dr]_{#2}\ar[rr]^{#3}&&#4\ar[dl]^{#5}\\&#6\ar@{}[u] |{\Longleftrightarrow}\ar@{}[u]|>>>>{#7}}}

\newcommand{\arr}[1]{\ar@<.5ex>[#1]\ar@<-.5ex>[#1]}
\newcommand{\arrr}[1]{\ar@<.7ex>[#1]\ar@<0ex>[#1]\ar@<-.7ex>[#1]}
\newcommand{\arrrr}[1]{\ar@<.9ex>[#1]\ar@<.3ex>[#1]\ar@<-.3ex>[#1]\ar@<-.9ex>[#1]}
\newcommand{\arrrrr}[1]{\ar@<1ex>[#1]\ar@<.5ex>[#1]\ar[#1]\ar@<-.5ex>[#1]\ar@<-1ex>[#1]}

\newcommand{\To}[1]{\xrightarrow{#1}}
\newcommand{\Too}[1]{\xrightarrow{\ \ #1\ \ }}
\newcommand{\From}[1]{\xleftarrow{#1}}

\newcommand{\Adjoint}[4]{\xymatrix@1{{#2}\ar@<.5ex>[r]^-{#1} &{#3} \ar@<.5ex>[l]^-{#4}}}

\def\id{\tn{id}}

\def\Cat{{\bf Cat}}
\def\Oprd{{\bf Oprd}}

\def\CMon{{\bf CMon}}

\def\Sets{{\bf Sets}}

\def\Rel{\mcR\tn{el}}
\def\JRel{J\mcR\tn{el}}

\def\Type{{\bf Type}}
\def\Set{{\bf Set}}

\def\sets{{\text \textendash}{\bf Alg}}

\def\dispInt{\parbox{.1in}{$\int$}}
\def\bhline{\Xhline{2\arrayrulewidth}}
\def\bbhline{\Xhline{2.5\arrayrulewidth}}

\def\mcE{\mc{E}}

\def\mcN{\mc{N}}
\def\mcO{\mc{O}}

\def\mcR{\mc{R}}
\def\mcS{\mc{S}}
\def\mcT{\mc{T}}

\def\bfS{{\bf S}}\def\bfSs{{\bf Ss}}
\def\bfT{{\bf T}}\def\bfTs{{\bf Ts}}

\def\tnN{\tn{N}}

\def\monOb{{\blacktriangle}}

\def\Fin{{\bf Fin}}
\def\PrO{{\bf PrO}}

\def\JLat{{\bf JLat}}
\def\JLats{{\bf JLats}}

\newcommand{\disc}[1]{{\ul{#1}}}

\makeatletter\let\c@figure\c@equation\makeatother %Aligns figure numbering and equation numbering.

\newtheorem{lemma}[subsubsection]{Lemma}
\newtheorem{proposition}[subsubsection]{Proposition}

\theoremstyle{remark}
\newtheorem{remark}[subsubsection]{Remark}
\newtheorem{example}[subsubsection]{Example}
\newtheorem{application}[subsubsection]{Application}
\newtheorem{warning}[subsubsection]{Warning}
\newtheorem{question}[subsubsection]{Question}
\newtheorem{guess}[subsubsection]{Guess}

\theoremstyle{definition}
\newtheorem{definition}[subsubsection]{Definition}
\newtheorem{notation}[subsubsection]{Notation}
\newtheorem{conjecture}[subsubsection]{Conjecture}

\usepackage{geometry}\newgeometry{left=1.7in,right=1.7in,top=1.4in,bottom=1.2in}

\begin{document}

\title[The operad of wiring diagrams]{The operad of wiring diagrams:\\
formalizing a graphical language for databases, recursion, and plug-and-play circuits}

\author{David I. Spivak}

\address{Department of Mathematics, Massachusetts Institute of Technology, Cambridge MA 02139}

\email{dspivak@math.mit.edu}

\thanks{This project was supported by ONR grants N000141010841 and N000141310260.}

\maketitle

\begin{abstract}

Wiring diagrams, as seen in digital circuits, can be nested hierarchically and thus have an aspect of self-similarity. We show that wiring diagrams form the morphisms of an operad $\mcT$, capturing this self-similarity. We discuss the algebra $\Rel$ of mathematical relations on $\mcT$, and in so doing use wiring diagrams as a graphical language with which to structure queries on relational databases. We give the example of circuit diagrams as a special case. We move on to show how plug-and-play devices and also recursion can be formulated in the operadic framework as well. Throughout we include many examples and figures.

\end{abstract}

\tableofcontents

\section{Introduction}

In the classic book {\em The Society of Mind} \cite{Min}, Marvin Minsky proposes that the human mind is composed of a society of agencies, each of which is composed of a society of smaller agencies, and so on. While the lower-level agencies could not be considered intelligent, their composition into structured wholes yields behavior that is markedly intelligent. He argues that this self-similar or recursive structure emerges naturally and is efficient for processing information. In this paper we offer a mathematical framework for modeling self-similar information processing structures of a certain kind.

In particular we will be interested in the self-similar nature of wiring diagrams, like those found in the design of digital circuits. One can wire together simple logic gates to form more complex ones, and wire these together to form still more complex circuits. Similarly, one may imagine that neurons can be wired together to form more complex information-processing centers in the brain, and that these can be wired together to form still more complex agencies. 

We will discuss the so-called {\em wiring diagrams operad} $\mcS$, which models the notion of wiring smaller entities together to form larger ones. We will be careful to introduce all the necessary operadic machinery, but it can be a bit dense for those who are unacquainted. For this reason, we will also draw plenty of pictures (such pictures constitute a large portion of the length of this document), which look like Figure \ref{fig:empty wiring diagram}.

\begin{figure}
\parbox{2.6in}{\includegraphics[height=2in]{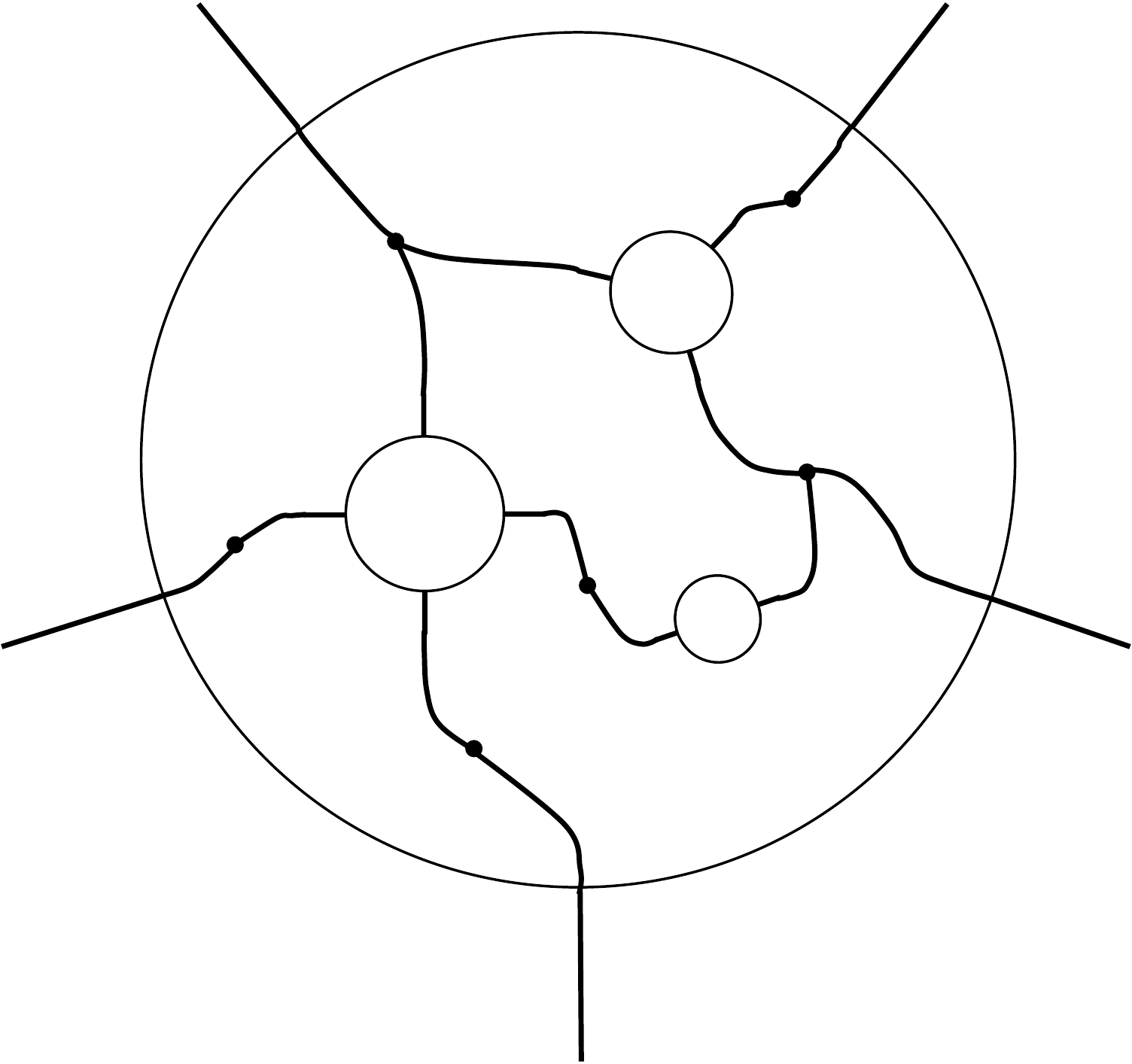}}
\caption{
An entity composed of three smaller entities, all wired together.
}\label{fig:empty wiring diagram}
\end{figure}

Operads were invented in the 1970s by \cite{BV} and \cite{May}. We will recall the definition of operads in Definition \ref{def:operad}, but it might be phrased in the current language using the slogan ``a wiring diagram of wiring diagrams is a wiring diagram". That is, if each little circle in Figure \ref{fig:empty wiring diagram} is filled in with its own little wiring diagram, the result will be a big outer circle with tiny ones wired together inside. This will be spelled out in Example \ref{ex:singly-typed wd}.

Deciding what kind of object each circle represents (whether it be a digital circuit, an agency in the brain, a mathematical relation, etc.), and deciding how a wiring diagram such as that in Figure \ref{fig:empty wiring diagram} produces a new such object from its constituent objects, is tantamount to choosing an {\em algebra} on the wiring diagrams operad. We will concentrate on the algebra $\Rel$ of mathematical relations, in which an outer relation is built by applying a wiring diagram to a set of constituent relations. This algebra $\Rel$ has immediate applications to visually representating relational database queries. 
\begin{align}\label{dia:3 queries}
\includegraphics[width=\textwidth]{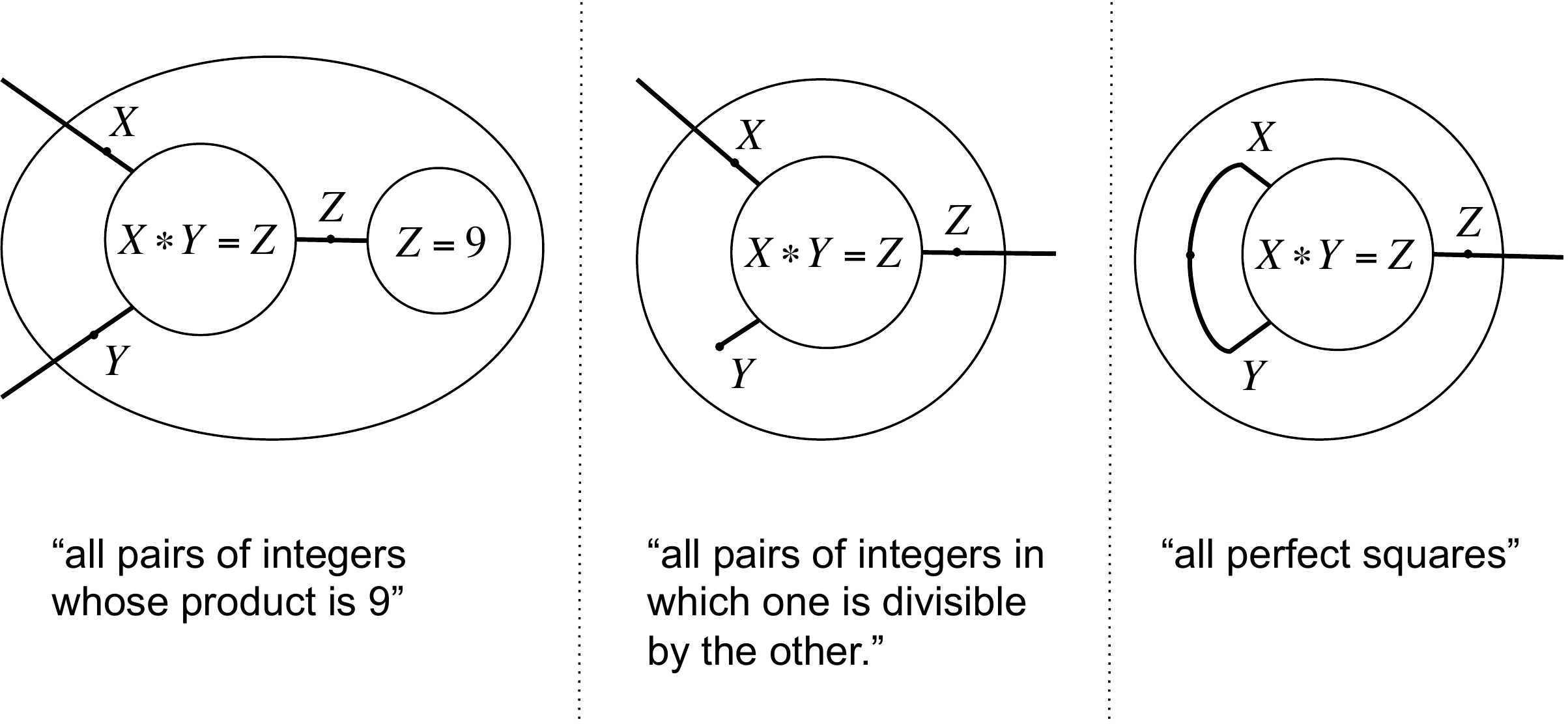}
\end{align}
But as the title of the paper suggests, one can also express recursion and so-called ``plug-and-play", or hot-plugging, in these terms. See Section \ref{sec:closed operads}.

Another reason for focusing on the algebra of relations is that digital circuits fall easily within this framework and provide a nice source of concrete examples. As a more abstract example, the algebra of relations can also serve as a crude model for agencies in the brain. We simply think of each agency as maintaining a relationship on its inputs. For example a person might be said to maintain a relationship between what she sees, hears, feels, and does; this can be represented as a 4-ary relation. Similarly, the visual cortex might be said to maintain relationships between the signals it receives from the eyes and from the rest of the brain. 

In \cite{JS}, it is shown that in the near future, cyber-physical systems will likely continue to blur the line between humans and the machines on which they work. One may imagine that the same hierarchical structures we use to conceive the human mind can be extended beyond the individual to encompass interactions between multiple humans and machines. Therefore, it is useful to have a single mathematical model, such as the operadic model presented throughout this paper, in which all such levels and layers can be discussed in a unified way.

More broadly, we hope this paper will introduce operads to a more general audience. The wiring diagrams operad discussed here is but a single possibility among many that may prove useful. Whenever one wants to discuss a situation in which structures are built out of other structures and the functionality follows suit, operads and their algebras may be employable.

\subsection{Related work}

There are of course many books on digital circuits, e.g. \cite{Eme}, and all that use wiring diagrams are related to the present paper, as we explain in Example \ref{ex:digital circuit}. Another related concept is that of the propagator network (\cite{RS}), in which computation is performed by a ``network of  autonomous but interconnected machines". Clearly, the underlying idea behind our work and theirs is quite similar: machines made of networked machines. This paper can be seen as providing a mathematical underpinning to this idea in terms of operads and their algebras. For example, no such distinction seems to be alluded to in the work of \cite{RS}.

Finally, there is a similarity between this work and previous work by the author, including \cite{Sp1}, \cite{Sp2}, and \cite{Sp3}. All of these provide different conceptions of databases. The one that appears closest to the usual relational model is the present paper, whereas the one that is most aligned with basic category theory is \cite{Sp3}. The category-theoretic intersection between these viewpoints seems to be that databases are nicely modeled by sheaves in various guises.

\subsection{Layout of this paper}

This paper will be structured as follows. In Section \ref{sec:basic defs}, we will define most of the necessary operadic machinery and introduce our main operad of interest, the singly-typed wiring diagrams operad $\mcS$. We will also discuss our main algebras of interest, the relational $\mcS$-algebras $\Rel_A$. Examples in this section include digital circuits and the three queries in (\ref{dia:3 queries}). Section \ref{sec:basic defs} constitutes the heart of the paper; the rest of the paper amounts to variations and explorations of its theme.

In Section \ref{sec:brief analysis} we discuss the notion of operad-theoretic invariants of the relational $\mcS$-algebra. After proving that certain attempts to find such invariants come up short, we conjecture that there are no nontrivial invariants (see Conjecture \ref{conj:quotient-free}). In Section \ref{sec:typed wd operad} we generalize the operad $\mcS$ to a typed version $\mcT$ in which each wire carries values of a specific type. We show how all the relational $\mcS$-algebras $\Rel_A$ extend to a single $\mcT$-algebra $\Rel$. In this section we more carefully discuss database queries (conjunctive and disjunctive), and compare to the wiring diagrams operad. In Section \ref{sec:closed operads} we show that $\mcS$ and $\mcT$ are in fact closed operads, and we give applications to recursion (in particular the factorial function) and the idea of plug-and-play devices. Finally, in Section \ref{sec:future work} we discuss possibilities for future work.

\subsection{Acknowledgments}

I'd like to thank my former UROP student Gaurav Singh, who was the first to write down a recursive function in terms of wiring diagrams. His diagram for the factorial function has a prominant role in this paper: see (\ref{dia:factorial}). I'd also like to thank Dave Balaban, Al Jones, Emmy Murphy, Eswaran Subrahmanian, and especially Nat Stapleton for many useful and interesting conversations.

\section{Operads and wiring diagrams}\label{sec:basic defs}

In this section we introduce operads, as well as our main operad of interest: the singly-typed wiring diagrams operad $\mcS$. We will also introduce algebras on an operad, as well as our main algebra of interest: the relational $\mcS$-algebra $\Rel$. 

We assume that the reader is familiar with basic category theory. Introductions to the subject include \cite{Mac} and \cite{Awo} for pure mathematicians, \cite{BW} and \cite{Pie} for computer scientists, and \cite{Sp1} for a general scientific audience. For those unacquainted with category theory, we hope that the diagrams will get across the basic idea.

\begin{notation}

We denote the category of sets by $\Set$ and the full subcategory of finite sets by $\Fin\ss\Set$.

If $n\in\NN$ is a natural number, we write $\disc{n}\in\Set$ to denote the set $\disc{n}:=\{1,2,\ldots,n\}$. We may also use symbols such as $n$ to denote finite sets, e.g. $n\in\Ob(\Fin)$. Given a finite set $n\in\Ob(\Fin)$, we denote its cardinality by $|n|\in\NN$. In other words, the following composite is the identity on $\NN$:
$$\NN\To{n\mapsto\disc{n}}\Ob(\Fin)\To{n\mapsto|n|}\NN.$$

For a set $X$, we let $\PP(X)=\{R\ss X\}$ denote its power set. We denote a one-element set by $\singleton$.

%If $X$ is a set, we may use the symbol $X$ to denote the corresponding discrete category, and in particular talk about functors $X\to\mcC$ for some category $\mcC$. For example, we have a functor $(x\mapsto\disc{x})\taking\NN\to\Fin$; here the set $\NN$ is being regarded as a discrete category. 

%If a category $\mcC$ has a terminal object $t$, then for any $A\in\mcC$ we denote the unique morphism $A\to t$ by $!_A$ or simply by $!$ if $A$ is clear from context. We denote the terminal category by $\disc{1}\in\Ob(\Cat)$.

\end{notation}

\subsection{Operads}

\begin{warning}\label{warn:operad}

Throughout this paper, we use the word {\em operad} to mean what is generally called a {\em symmetric colored operad} or a {\em symmetric multicategory}. This abbreviated nomenclature is not new, for example it is used in \cite{Lur}. Hopefully no confusion will arise. For a full treatment of operads, multicategories, and how they fit into a larger mathematical context, see \cite{Lei}.

\end{warning}

\begin{definition}\label{def:operad}

An {\em operad} $\mcO$ consists of some constituents (A. objects, B. morphisms, C. identities, D. compositions) conforming to some laws (1. identity law, 2. associativity law). The required constituents are as follows.
\begin{enumerate}[\hsp A.]
\item There is a collection $\Ob(\mcO)$, each element of which is called an {\em object} of $\mcO$.
\item For each object $y\in\Ob(\mcO)$, finite set $n\in\Ob(\Fin)$, and $n$-indexed set of objects $x\taking n\to\Ob(\mcO)$, there is a set $\mcO_n(x;y)\in\Ob(\Set)$. Its elements are called {\em morphisms from $x$ to $y$} in $\mcO$. 
\item For every object $x\in\Ob(\mcO)$, there is a specified morphism denoted $\id_x\in\mcO_1(x;x)$ called {\em the identity morphism on $x$}.
\item Let $s\taking m\to n$ be a morphism in $\Fin$. Let $z\in\Ob(\mcO)$ be an object, let $y\taking n\to\Ob(\mcO)$ be an $n$-indexed set of objects, and let $x\taking m\to\Ob(\mcO)$ be an $m$-indexed set of objects. For each element $i\in n$, write $m_i:=s^\m1(i)$ for the pre-image of $s$ under $i$, and write $x_i=x|_{m_i}\taking m_i\to\Ob(\mcO)$ for the restriction of $x$ to $m_i$. Then there is a function 
\begin{align}\label{dia:composition formula}
\circ\taking\mcO_n(y;z)\times\prod_{i\in n}\mcO_{m_i}(x_i;y(i))\too\mcO_{m}(x;z),
\end{align} 
called {\em the composition formula}.
\end{enumerate}
Given an $n$-indexed set of objects $x\taking n\to\Ob(\mcO)$ and an object $y\in\Ob(\mcO)$, we sometimes abuse notation and denote the set of morphisms from $x$ to $y$ by $\mcO(x_1,\ldots,x_n;y)$.
\footnote{There are three abuses of notation when writing $\mcO(x_1,\ldots,x_n;y)$, which we will fix one by one. First, it confuses the set $n\in\Ob(\Fin)$ with its cardinality $|n|\in\NN$. But rather than writing $\mcO(x_1,\ldots,x_{|n|};y)$, it would be more consistent to write $\mcO(x(1),\ldots,x(|n|);y)$, because we have assigned subscripts another meaning in part D. But even this notation unfoundedly suggests that the set $n$ has been endowed with a linear ordering, which it has not. This may be seen as a more serious abuse, but see Remark \ref{rem:symmetry}.}
We may write $\Hom_\mcO(x_1,\ldots,x_n;y)$, in place of $\mcO(x_1,\ldots,x_n;y)$, when convenient. We can denote a morphism $\phi\in\mcO_n(x;y)$ by $\phi\taking x\to y$ or by $\phi\taking (x_1,\ldots,x_n)\to y$; we say that each $x_i$ is a {\em domain object} of $\phi$ and that $y$ is the {\em codomain object} of $\phi$. We use infix notation for the composition formula, e.g. writing $\psi\circ(\phi_1,\ldots,\phi_n)$.

Here are the laws that govern the above constituents:
\begin{enumerate}[\hsp 1.]
\item for every $x_1,\ldots,x_n,y\in\Ob(\mcO)$ and every morphism $\phi\taking(x_1,\ldots,x_n)\to y$, we have
$$\phi\circ(\id_{x_1},\ldots,\id_{x_n})=\phi\hsp\tn{and}\hsp\id_y\circ\phi=\phi;$$ \hsp{and};
\item an appropriate associative law.
\footnote{Here is the associative law. Let $m\To{s}n\To{t}p$ be composable morphisms in $\Fin$. Let $z\in\Ob(\mcO)$ be an object, let $y\taking p\to\Ob(\mcO)$, $x\taking n\to\Ob(\mcO)$, and $w\taking m\to\Ob(\mcO)$ respectively be a $p$-indexed, $n$-indexed, and $m$-indexed set of objects. For each $i\in p$, write $n_i=t^\m1(i)$ for the pre-image and $x_i\taking n_i\to\Ob(\mcO)$ for the restriction. Similarly, for each $k\in n$ write $m_k=s^\m1(k)$ and $w_k\taking m_k\to\Ob(\mcO)$; for each $i\in p$, write $m_{i,-}=(t\circ s)^\m1(i)$ and $w_{i,-}\taking m_{i,-}\to\Ob(\mcO)$; for each $j\in n_i$, write $m_{i,j}:=s^\m1(j)$ and $w_{i,j}\taking m_{i,j}\to\Ob(\mcO)$. Then the diagram below commutes:
$$\xymatrix@=18pt{
&{\hspace{.8in}}\save[]+<0cm,0cm>*\txt<30pc>{$
\mcO_p(y;z)\times\prod_{i\in p}\mcO_{n_i}(x_i;y(i))\times\prod_{i\in p,\ j\in n_i}\mcO_{m_{i,j}}(w_{i,j};x_i(j))$}
\ar[rd]\ar[ld]\restore\\
\mcO_n(x;z)\times\prod_{k\in n}\mcO_{m_k}(w_k;x(k))\ar[dr]&&\mcO_p(y;z)\times\prod_{i\in p}\mcO_{m_{i,-}}(w_{i,-};y(i))\ar[dl]\\
&\mcO_m(w;z)
}
$$
}

\end{enumerate}

\end{definition}

\begin{remark}\label{rem:symmetry}

In this remark we will discuss the abuse of notation in Definition \ref{def:operad} and how it relates to an action of a symmetric group on each morphism set in our definition of operad. We follow the notation of Definition \ref{def:operad}, especially following the use of subscripts in the composition formula.

Suppose that $\mcO$ is an operad, $z\in\Ob(\mcO)$ is an object, $y\taking n\to\Ob(\mcO)$ is an $n$-indexed set of objects, and $\phi\taking y\to z$ is a morphism. If we linearly order $n$, enabling us to write $\phi\taking (y(1),\ldots,y(|n|))\to z$, then changing the linear ordering amounts to finding an isomorphism of finite sets $\sigma\taking m\To{\iso} n$, where $|m|=|n|$. Let $x=y\circ\sigma$ and for each $i\in n$, note that $m_i=\sigma^\m1(\{i\})=\{\sigma^\m1(i)\}$, so $x_i=x|_{\sigma^\m1(i)}=y(i)$. Taking $\id_{x_i}\in\mcO_{m_i}(x_i;y(i))$ for each $i\in n$, and using the identity law, we find that the composition formula induces a bijection $\mcO_n(y;z)\To{\iso}\mcO_m(x;z)$, which we might denote by 
$$\sigma\taking\mcO(y(1),y(2),\ldots,y(n);z)\iso\mcO\big(y(\sigma(1)),y(\sigma(2)),\ldots,y(\sigma(n));z\big).$$
In other words, there is an induced group action of $\Aut(n)$ on $\mcO_n(x;z)$, where $\Aut(n)$ is the group of permutations of an $n$-element set.

Throughout this paper, we will permit ourselves to abuse notation and speak of morphisms $\phi\taking (x_1,x_2,\ldots,x_n)\to y$ for a natural number $n\in\NN$, without mentioning the abuse inherent in choosing an order, so long as it is clear that permuting the order of indices would not change anything up to canonical isomorphism.

\end{remark}

\begin{example}\label{ex:Sets}

We define the operad of sets, denoted $\Sets$, as follows. We put $\Ob(\Sets):=\Ob(\Set)$. Given a natural number $n\in\NN$ and objects $X_1,\ldots,X_n, Y\in\Ob(\Sets)$, we define 
$$\Sets(X_1,X_2,\ldots,X_n;Y):=\Hom_\Set(X_1\times X_2\times\cdots\times X_n,Y).$$ 
For any $X\in\Ob(\Sets)$ the identity morphism $\id_X\taking X\to X$ is the same identity as that in $\Set$. 

The composition formula is as follows. Suppose given a set $Z\in\Ob(\Set)$, a natural number $n\in\NN$, sets $Y_1,\ldots,Y_n\in\Ob(\Set)$, and for each $i$ natural numbers $m_1,\ldots,m_n\in\NN$ and sets $X_{i,1},\ldots,X_{i,m_i}\in\Ob(\Set)$. Suppose furthermore that we have composable morphisms: a function $g\taking Y_1\times\cdots\times Y_n\to Z$ and for each $i\in\disc{n}$ a function $f_i\taking X_{i_1}\times\ldots\times X_{i,m_i}\to Y_i$. We need a function $X_{1,1}\times X_{1,2}\times\cdots\times X_{n,m_n}\to Z$, which we take to be the composite $$(X_{1,1}\times\cdots\times X_{1,m_1})\times\cdots\times(X_{n,1}\times\cdots\times X_{n,m_n})\To{\;f_1\times\cdots\times f_n\;}Y_1\times\cdots\times Y_n\To{\;g\;}Z.$$

It is not hard to see that this composition formula is associative.

\end{example}

\begin{example}\label{ex:Ss}

Let $\bfS$ be any category with finite products. Such categories include: the category $\Fin$ of finite sets, the category $\PrO$ of preorders, the category $\JLat$ of \href{http://en.wikipedia.org/wiki/Semilattice}{\em join-semilattices}, the category $\Type$ of types in a programming language, the category $\Cat$ of small categories, etc. We then define an operad $\bfSs$ with $$\Ob(\bfSs)=\Ob(\bfS)\hsp\tn{and}\hsp\bfSs_n(X_1,\ldots,X_n;Y)=\Hom_\bfS(X_1\times\cdots\times X_n,Y).$$ The composition formula is similar to that discussed in Example \ref{ex:Sets}.

\end{example}

\begin{example}\label{ex:commutative operad}

The {\em commutative operad} $\mcE$ has one object, say $\Ob(\mcE)=\{\monOb\}$, and for each $n\in\NN$ it has a single $n$-ary morphism, $\mcE_n(\blacktriangle,\ldots,\blacktriangle;\blacktriangle)=\{\mu_n\}$. In fact $\mcE$ can be viewed as a special case of Example \ref{ex:Ss} where $\bfS=\disc{1}$ is the terminal category; i.e. $\mcE\iso\disc{1}{\bf s}$. We will see in Lemma \ref{lemma:terminal operad} that $\mcE$ is the terminal object in the category of operads. (Again, recall from Warning \ref{warn:operad} that all of our operads are actually symmetric colored operads, and we are saying that $\mcE$ is the terminal object in that category.)

\end{example}

\begin{example}[{\bf Main Operad Example}]\label{ex:singly-typed wd}

In this example we discuss the operad $\mcS$ of singly-typed wiring diagrams. The objects of $\mcS$ are simply finite sets, $\Ob(\mcS)=\Ob(\Fin)$. However, we refer to each object as a {\em singly-typed star} and we draw an object, say $X:=\{a,b,c,d,e\}\in\Ob(\mcS)$, as follows:\\

\begin{align}\label{dia:unlabeledOb}
\parbox{1.35in}{\includegraphics[width=1.35in]{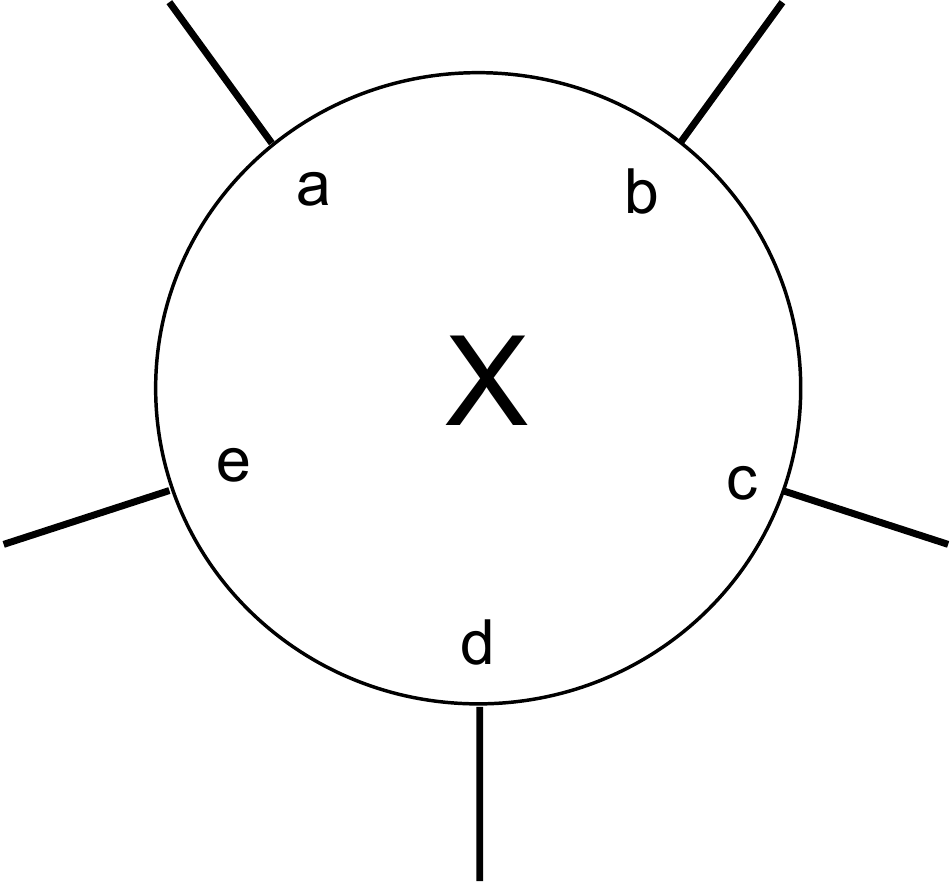}}
\end{align}
The order of the elements around the circle does not matter: $X$ is just a set. Each element of $X$ is called a {\em wire of $X$}.

Let $Y\in\Ob(\mcS)$ be an object, let $n\in\NN$ be a natural number, and let $X_1,\ldots,X_n\in\Ob(\mcS)$ be singly-typed stars. A morphism $\phi\taking X_1,\ldots,X_n\to Y$ consists of a cospan, 
\footnote{In fact, a morphism $\phi\taking(X_1,\ldots,X_n)\to Y$ in $\mcS$ is defined to be a cospan {\em up to isomorphism}. In other words, given a commutative diagram as below,
$$\xymatrix{
&C\ar[dd]^\iso\\
X_1\amalg\cdots\amalg X_n\ar[ur]^f\ar[dr]_{f'}&&Y\ar[ul]_{g}\ar[dl]^{g'}\\
&C'}
$$
we say that the cospans $(f,g)$ and $(f',g')$ {\em represent} the same morphism $X_1,\ldots,X_n\To{\phi}Y$. This just means that two wiring diagrams are considered to be the same morphism in $\mcS$ if the only difference between them is a renaming of cables.
}
i.e. a finite set $C\in\Ob(\Fin)$ and functions 
$$\xymatrix{
&Y\ar[d]^g\\
X_1\amalg\cdots\amalg X_n\ar[r]_-f&C.
}
$$ 
We refer to such a morphism as an {\em singly-typed wiring diagram} (or a {\em wiring diagram} or even a {\em WD}, for short); we refer to the domain objects $X_1,\ldots,X_n$ as the {\em inner stars}, the codomain object $Y$ as the {\em outer star}, and the set $C$ as the {\em set of cables}. For any wire $y\in Y$ with $g(y)=c$ (respectively $x\in X_i$ with $f(x)=c$), we say that $y$ (respectively $x$) is {\em soldered onto cable $c$}. Figure \ref{fig:wiring diagram} should make all this quite a bit clearer.
\begin{figure}[h]
\begin{align}\label{dia:unlabeledMorphism}
\parbox{3.1in}{\includegraphics[width=3.1in]{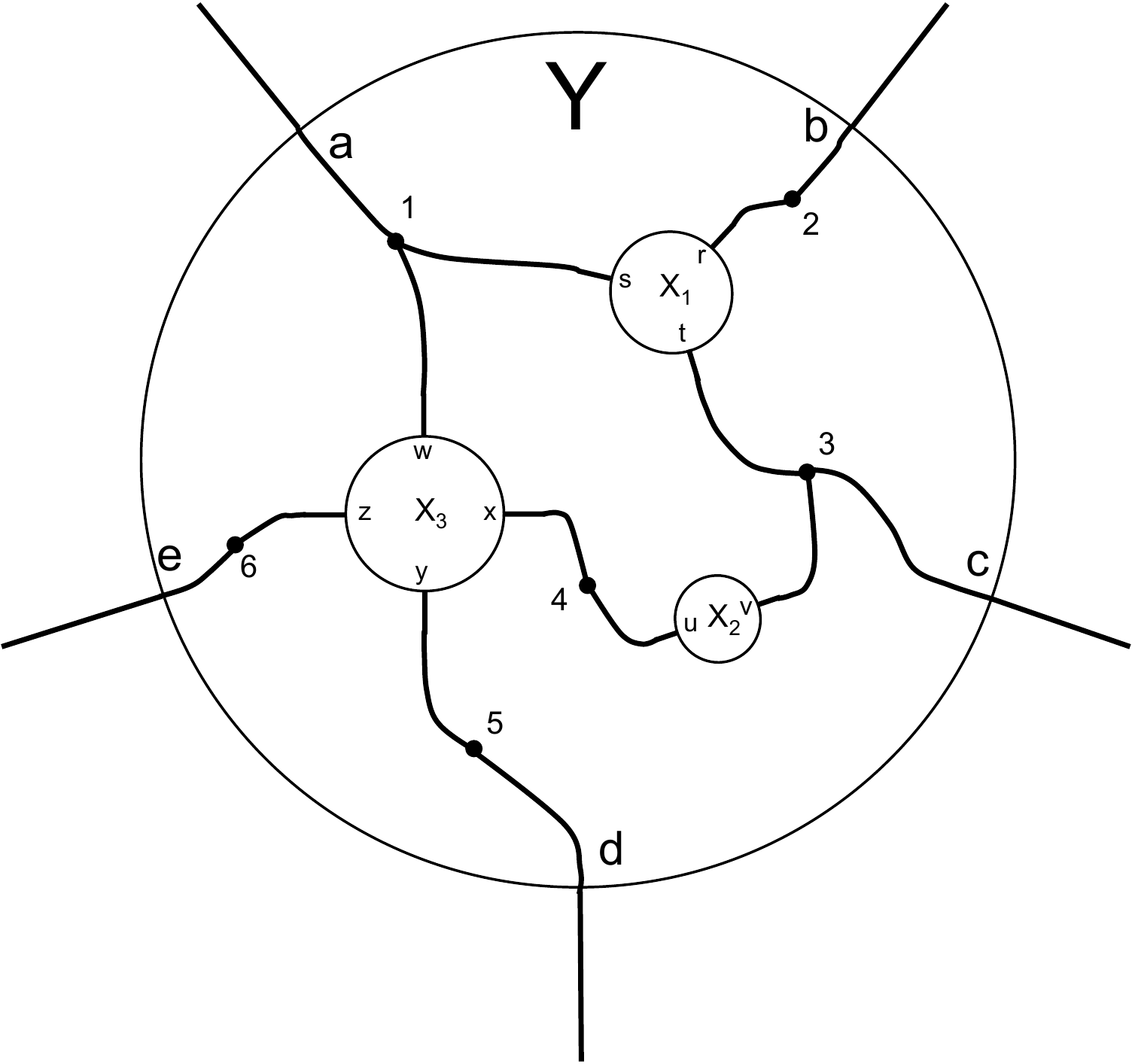}}
\end{align}
\caption{
Here is a picture of $\phi\taking (X_1,X_2,X_3)\to Y$. The codomain of $\phi$ is the outer star $Y=\{a,b,c,d,e\}$ and the $n=3$ domain objects are the inner stars $X_1=\{r,s,t\}, X_2=\{u,v\},$ and $X_3=\{w,x,y,z\}$. The morphism $\phi$ has cables $C=\{1,2,3,4,5,6\}$, and our picture shows, for example, that $f_1(s)=f_3(w)=g(a)=1$, i.e. that wires $s,w,$ and $a$ are soldered onto cable $1$.
}\label{fig:wiring diagram}
\end{figure}

Given composable morphisms $\phi'\taking Y_1,\ldots,Y_n\to Z$ and $\phi_1,\ldots,\phi_n$, where each $\phi_i\taking X_{i,1},\ldots,X_{i,m_i}\to Y_i$, we can arrange their representing cospans in the form of the W-shaped diagram below. The composition formula is given by taking the pushout $P$ of the middle square.
\begin{align}\label{dia:untyped composition}
\xymatrix@=35pt{&&Z\ar[d]^{g'}\\
&\coprod_{i\in n}Y_i\ar[r]^{f'}\ar[d]_{\amalg_i g_i}&D\ar[d]\\
\coprod_{i\in n}\coprod_{j\in m_i}X_{i,j}\ar[r]_-{\amalg_i f_i}&\coprod_{i\in n}C_i\ar[r]&P\lrlimit
}
\end{align}
Here $\phi'\circ(\phi_1,\ldots,\phi_n)$ is represented by the cospan $\amalg_{i\in n}\amalg_{j\in m_i}X_{i,j}\to P\from Z$.

The composition formula for $\phi'\circ(\phi_1,\ldots,\phi_n)$ is intuitive when drawing pictures. We begin with wiring diagrams $\phi_i$ of little inner stars inside intermediary stars and a wiring diagram $\phi'$ of intermediary stars inside a big outer star. The wires of the intermediary stars $Y_i$ will match interior cables to exterior cables. The composition is visualized by first placing the little WDs $\phi_i$ into their respective outer stars $Y_i$, and then removing the intermediary circles $Y_i$, while preserving the connections they made. This is captured mathematically by the pushout in (\ref{dia:untyped composition}) and is drawn below.
\begin{align}\label{dia:untyped composition drawing}
\parbox{\textwidth}{\includegraphics[width=\textwidth]{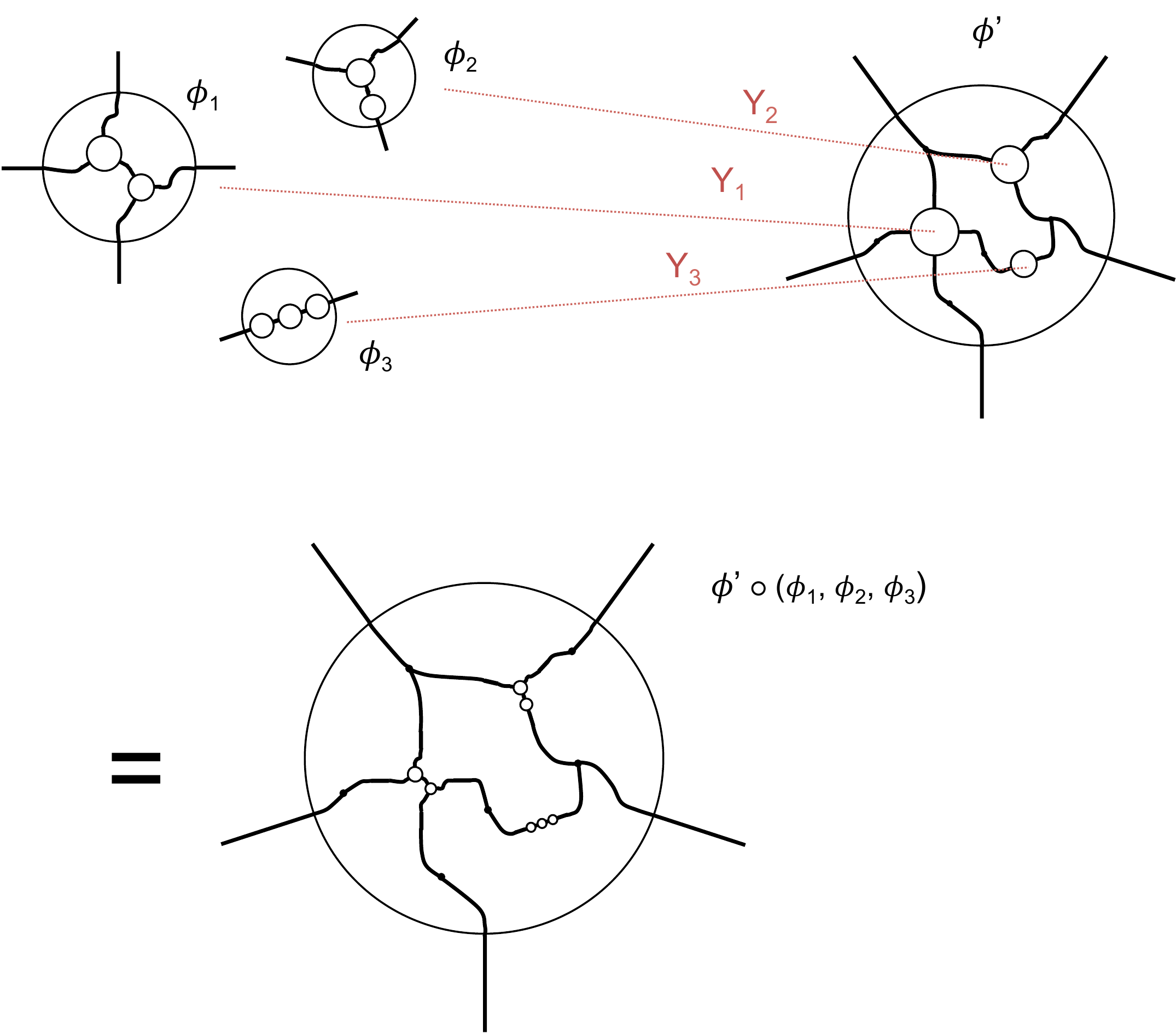}}
\end{align}
The identity and associativity law hold by the usual properties of pushouts.

\end{example}

The idea of this paper is to show several applications of the singly-typed wiring diagrams operad as well as a variant, the typed wiring diagrams operad, discussed in Section \ref{sec:typed wd operad}. First we must explain the notion of algebras on an operad.

\subsection{Operad functors and algebras on an operad}

An algebra on an operad is a special case of a functor between operads. 

\begin{definition}

Let $\mcO$ and $\mcO'$ be operads. An {\em operad functor (\tn{or simply} functor) $F$ from $\mcO$ to $\mcO'$}, denoted $F\taking\mcO\to\mcO'$, consists of some constituents (A. map on objects, B. map on morphisms) conforming to some laws (1. identity law, 2. composition law). The required constituents for $F$ are as follows:
\begin{enumerate}[\hsp A.]
\item There is a function $\Ob(F)\taking\Ob(\mcO)\to\Ob(\mcO')$.
\item For each object $y\in\Ob(\mcO)$, finite set $n\in\Ob(\Fin)$, and $n$-indexed set of objects $x\taking n\to\Ob(\mcO)$, there is a function $$F_n\taking\mcO_n(x;y)\to\mcO'_n(Fx;Fy).$$
\end{enumerate}
As in B. above, we often denote $\Ob(F)$, and also each $F_n$, simply by $F$. The laws that govern these constituents are as follows:
\begin{enumerate}[\hsp 1.]
\item For each object $x\in\Ob(\mcO)$, the equation $F(\id_x)=\id_{Fx}$ holds.
\item Let $s\taking m\to n$ be a morphism in $\Fin$. Let $z\in\Ob(\mcO)$ be an object, let $y\taking n\to\Ob(\mcO)$ be an $n$-indexed set of objects, and let $x\taking m\to\Ob(\mcO)$ be an $m$-indexed set of objects. Then, with notation as in Definition \ref{def:operad}, the following diagram of sets commutes:
\begin{align}\label{dia:operad functor on composition}
\xymatrix{
\mcO_n(y;z)\times\prod_{i\in n}\mcO_{m_i}(x_i;y(i))\ar[r]^-F\ar[d]_\circ&
\mcO'_n(Fy;Fz)\times\prod_{i\in n}\mcO'_{m_i}(Fx_i;Fy(i))\ar[d]^\circ\\
\mcO_m(x;z)\ar[r]_F&\mcO'_m(Fx;Fz)
}
\end{align}
\end{enumerate}

We denote the category of operads and operad functors by $\Oprd$.

\end{definition}

\begin{warning}

What we call operad functors are usually called {\em operad morphisms} or {\em morphisms of operads} (see \cite[Chapter 2]{Lei}). We found that the notion of morphisms {\em in an operad} and that of morphisms {\em of operads} were too easily confused using this terminology. The confusion seems analogous to what would arise if we dismissed the term ``functor" altogether from category theory, relying on the term ``category morphism" or ``morphism of categories" to take over and fulfill its role.
 
 \end{warning}
 
\begin{example}\label{ex:operads s}

Suppose that $\bfS$ and $\bfT$ are categories with all finite products. By Example \ref{ex:Ss} we have operads $\bfSs$ and $\bfT{\bf s}$. If $f\taking\bfS\to\bfT$ is a functor that preserves finite products then there is an induced operad functor $F\taking\bfSs\to\bfTs$. Indeed, on objects we can put $F=f$ since $\Ob(\bfSs)=\Ob(\bfS)$ and $\Ob(\bfTs)=\Ob(\bfT)$. Every morphism $\phi\taking (x_1,\ldots,x_n;y)$ in $\bfS$ is given by a morphism $\phi\taking x_1\times\cdots\times x_n\to y$ in $\bfS$, so we put 
$$F(\phi)\taking F(x_1)\times\cdots\times F(x_n)\iso F(x_1\times\cdots\times x_n)\Too{f(\phi)}F(y).$$

\end{example}

\begin{lemma}\label{lemma:terminal operad}

Let $\mcO$ be an operad, and let $\mcE$ be the commutative operad from Example \ref{ex:commutative operad}. Then there is a unique operad functor $\mcO\to\mcE$. In other words, $\mcE$ is a terminal object in $\Oprd$.

\end{lemma}

\begin{proof}
The operad $\mcE$ has one object and one morphism in each $\Hom$-set, so the proof is straightforward.
\end{proof}

\begin{definition}\label{def:algebras on operad}

Let $\mcO$ be an operad. An {\em algebra on $\mcO$} is an operad functor $R\taking\mcO\to\Sets$. 

Given two algebras $R,S\taking\mcO\to\Sets$, an {\em algebra morphism $p$ from $R$ to $S$}, denoted $p\taking R\to S$, consists of one constituent (A. components) and one law (1. naturality) as follows.
\begin{enumerate}[\hsp A.]
\item For each object $x\in\Ob(\mcO)$ there is a function $p_x\taking R(x)\to S(x)$.
\end{enumerate}
The law that governs the components is as follows:
\begin{enumerate}[\hsp 1.]
\item Suppose given an object $y\in\Ob(\mcO)$, a finite set $n\in\Ob(\Fin)$, an $n$-indexed set of objects $x\taking n\to\Ob(\mcO)$, and a morphism $\phi\taking x\to y$. Then the following diagram of sets (the so-called {\em naturality square}) commutes:
$$
\xymatrix{
\prod_{i\in n}R(x_i)\ar[rr]^-{\prod_{i\in n}p_{x_i}}\ar[d]_{R(\phi)}&&\prod_{i\in n}S(x_i)\ar[d]^{S(\phi)}\\
R(y)\ar[rr]_-{p_y}&&S(y).
}
$$
\end{enumerate}

The algebras on $\mcO$ and their morphisms form a category, which we denote by $\mcO\sets$.

\end{definition}

\begin{remark}

Definition \ref{def:algebras on operad} defines algebra morphisms, but it could easily be modified to define operad natural transformations between arbitrary operad functors $F,G\taking\mcO\to\mcO'$. Operad natural transformations $\alpha\taking F\to G$ are defined in the usual way: a component $\alpha_x$ for each object $x\in\Ob(\mcO)$ and naturality squares for each morphism in $\mcO$. We will not explicitly give the definition in this paper, but it straightforwardly generalizes Definition \ref{def:algebras on operad}, in which $\mcO'=\Sets$.

\end{remark}

\begin{example}\label{ex:commutative monoids}

Let $\mcE$ be the commutative operad of Example \ref{ex:commutative operad}. An $\mcE$-algebra $S\taking\mcE\to\Sets$ consists of a set $M\in\Ob(\Set)$, and for each natural number $n\in\NN$ a morphism $\mu_n\taking M^n\to M$. It is not hard to see that, together the morphism $\mu_2\taking M\times M\to M$ and the element $\mu_0\taking\singleton\to M$ give $M$ the structure of a commutative monoid. Indeed, the associativity and unit axioms are encoded in the axioms for operads and their morphisms. The commutativity of multiplication arises by applying the commutative diagram (\ref{dia:operad functor on composition}) to the symmetric group action as discussed in Remark \ref{rem:symmetry}.

\end{example}

\begin{lemma}

Let $\mcO$ be an operad. The category $\mcO\sets$ has all limits and colimits. In particular the initial object sends each $x\in\Ob(\mcO)$ to the empty set, and a terminal object (denoted $\singleton^\mcO\taking\mcO\to\Sets$) sends each $x\in\Ob(\mcO)$ to a singleton set.

\end{lemma}

\begin{proof}

Limits and colimits of $\mcO$-algebras are formed object-wise on $\mcO$.

\end{proof}

The following lemma allows us to transfer algebras on one operad to algebras on another.

\begin{lemma}\label{lemma:pullback of algebras}

Suppose that $F\taking\mcO\to\mcO'$ is an operad functor. It induces a functor $F^*\taking\mcO'\sets\to\mcO\sets$, sending an $\mcO'$-algebra $i\taking\mcO'\to\Sets$ to its composition with $F$. This functor has a left adjoint,
$$\xymatrix{F_!\taking\mcO\sets\ar@<.5ex>[r]&\mcO'\sets:\hspace{-2.5pt}F^*\ar@<.5ex>[l]}$$

\end{lemma}

\begin{proof}[Sketch of proof]
The functor $F^*\taking\mcO'\sets\to\mcO\sets$ is clear. The existence of a left adjoint (also known as a left Kan extension) is well-known, but hard to find in the literature. (See \cite{Lur} for a $\infty$-categorical version.) The construction of the operadic left Kan extension $F_!$ proceeds analogously to the construction of left Kan extensions in basic category theory. Because we will not use the left adjoint in the current paper, we leave its existence unproven. 

%2013/04/26 Left Kan extension proof sketch

\end{proof}

In this paper we are most concerned with wiring diagrams, and Example \ref{ex:singly typed relational alg} discusses the algebra of relations on the singly-typed wiring diagrams operad $\mcS$.

\begin{example}[{\bf Main Algebra Example}]\label{ex:singly typed relational alg}

Let $\mcS$ be the singly-typed wiring diagram operad discussed in Example \ref{ex:singly-typed wd}, and fix a set $A\in\Ob(\Set)$. In this example we associate to $A$ an $\mcS$-algebra $\Rel_A$. We call it {\em the relational $\mcS$-algebra of type $A$}.

To specify this $\mcS$-algebra, i.e. this operad functor $\Rel_A\taking\mcS\to\Sets$, we need to specify it on objects and on morphisms. Let $X\in\Ob(\mcS)=\Ob(\Fin)$ be a singly-typed star; we assign 
\begin{align}\label{dia:rel definition}\Rel_A(X):=\{R\ss A^X\}=\PP(A^X).\end{align}
That is, $\Rel_A(X)$ is the set of $|X|$-ary relations on $A$. For any set $A$, if $X=\emptyset$ then $\Rel_A(X)=\PP(\singleton)$ has two elements and can be considered to be boolean valued.

Now suppose that $\phi\taking X_1,\ldots,X_n\to Y$ is a singly-typed wiring diagram given by the cospan to the left below
$$\xymatrix{
&Y\ar[d]^g\\
X_1\amalg\cdots\amalg X_n\ar[r]^-f&C.
}
\hspace{.8in}
\xymatrix{
&A^Y\\
A^{X_1\amalg\cdots\amalg X_n}&A^C\ar[l]_-{A^f}\ar[u]_{A^g}
}
$$  
Note that it induces the diagram to the right above. We need to specify a function 
$$\Rel_A(\phi)\taking\Rel_A(X_1)\times\cdots\times\Rel_A(X_n)\to\Rel_A(Y).$$ 
So suppose given relations $R_1\ss A^{X_1},\ldots,R_n\ss A^{X_n}$. Then $\Rel_A(\phi)(R_1,\ldots,R_n)$ is the relation $S\ss A^Y$ defined by first forming the pullback $Q$ in the diagram below, 
\begin{align}\label{dia:composition formula for singly typed}
\xymatrix{
&A^Y\\
A^{X_1}\times\cdots\times A^{X_n}&A^C\ar[u]_{A^g}\ar[l]_-{A^f}&S\ar@{^(->}@/_1pc/[ul]\\
R_1\times\cdots\times R_n\ar@{^(->}[u]&Q\lrlimit\ar[l]\ar[u]\ar@{->>}@/_1pc/[ur]
}
\end{align}
and then taking as $S$ the image of $Q$ under the induced function $Q\to A^Y$. We refer to $\Rel_A(\phi)(R_1,\ldots,R_n)$ as the {\em $\phi$-conjunction} of relations $R_1,\ldots,R_n$.

Because epi-mono factorizations in $\Set$ are unique up to isomorphism, it is easy to show that, as specified, $\Rel_A$ respects the composition formula (\ref{dia:untyped composition}) in $\mcS$.

\end{example}

\begin{example}\label{ex:digital circuit}

In the context of digital circuits, it is well known that one can begin with NAND gates and put them together to build NOT gates, AND gates, OR gates, etc. A NAND gate is an element of $\Rel_{\{True,False\}}(\{\tn{A,B,out}\})$, namely the relation represented by the table
$$\parbox{1.5in}{
\begin{tabular}{| l | l | l |}
\bhline
\multicolumn{3}{|c|}{NAND}\\\bhline
{\bf A}&{\bf B}&{\bf out}\\\bbhline
True&True&False\\
True&False&True\\
False&True&True\\
False&False&True\\\bhline
\end{tabular}}
\hspace{.7in}
\parbox{1.5in}{
\includegraphics[width=1.5in]{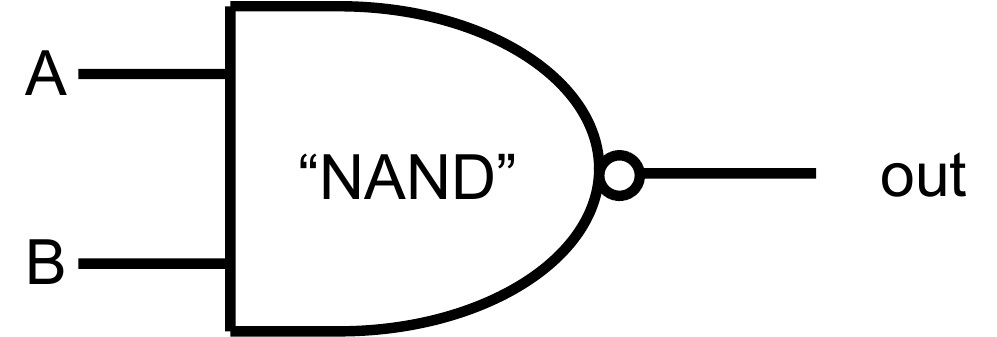}
}
$$
The wiring diagram $\phi$ given by the cospan $\{\tn{A,B,out}\}\To{f}\{\tn{in,out}\}\From{g}\{\tn{in,out}\}$, where $f(\tn{A})=f(\tn{B})=\tn{in}$, is depicted by the form of the diagram below. (The interior is ``grayed out" for aesthetic reasons, to invoke the idea of viewing NOT as a black box.)
\begin{center}
\includegraphics[width=2in]{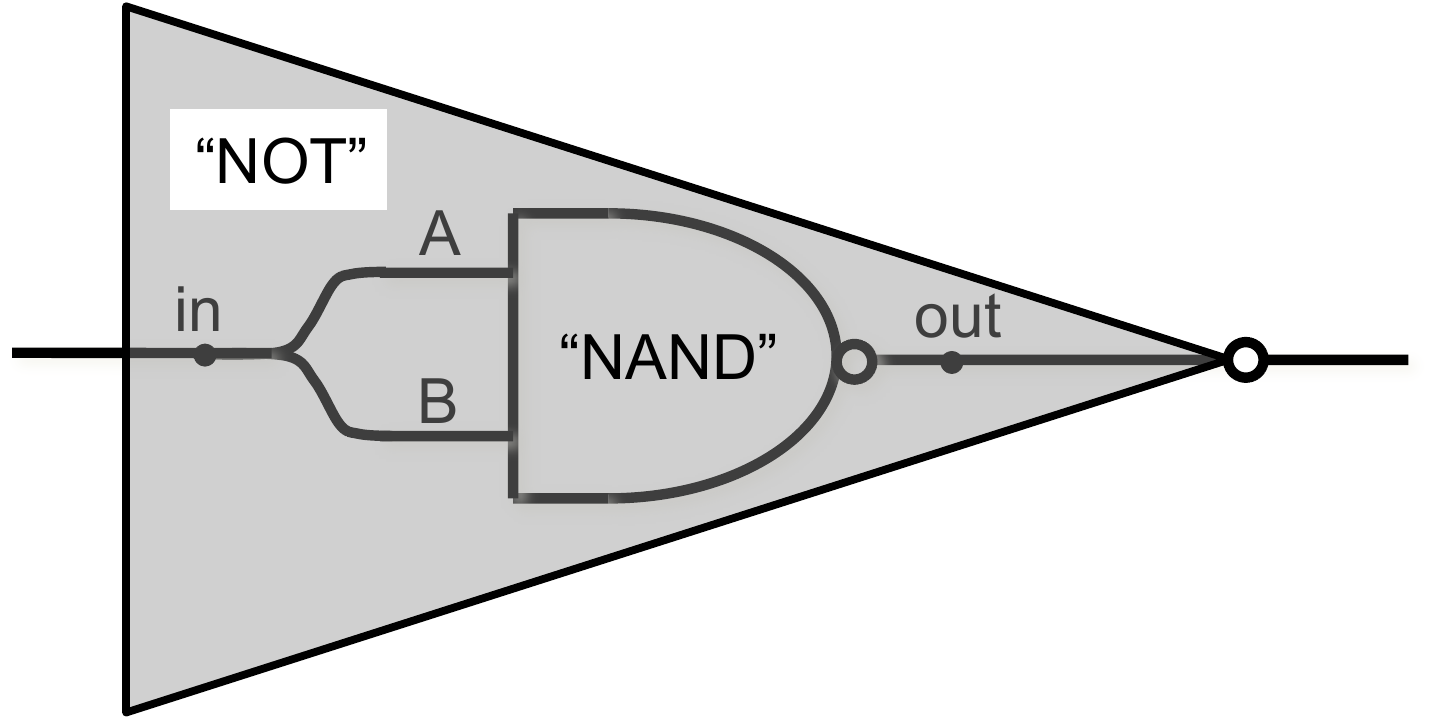}
\end{center}
Note however that this diagram depicts more than a morphism $\phi\taking\{\tn{A,B,out}\}\to \{\tn{in,out}\}$ in $\mcS$, but also its application to NAND. That is, it is a depiction of 
$$\tn{NOT}:=\Rel_{\{True,False\}}(\phi)(\tn{NAND}),$$
as an element of $\Rel_{\{True,False\}}(\{\tn{in,out}\})$. One can form relations AND, OR, etc., similarly, and then move on to so-called half adders, full adders, $n$-bit adders, etc. 

At each point the operad allows us to hide the inner implementation details if we choose, or use the composition formula to show the whole picture. We can also regroup different parts of the circuit, and the associative law of operads assures us that this will not change the result.

\end{example}

\begin{example}\label{ex:three queries revisited}

In this example we return to the three queries from the introduction; see (\ref{dia:3 queries}). The wiring diagrams for these all involve a three-wire star $\{X,Y,Z\}$, and the first diagram also involves a one-wire star $\{Z\}$. 
\begin{center}
\includegraphics[width=\textwidth]{threeQueries}
\end{center}
$$\tiny
\xymatrix{
&\{X,Y\}\ar[d]^{g_1}\\
\{X,Y,Z\}\amalg\{Z\}\ar[r]_{f_1}&\{X,Y,Z\}
}\hsp
\xymatrix{
&\{X,Z\}\ar[d]^{g_2}\\
\{X,Y,Z\}\ar[r]_{f_2}&\{X,Y,Z\}
}\hsp
\xymatrix{
&\{Z\}\ar[d]^{g_3}\\
\{X,Y,Z\}\ar[r]_{f_3}&\{XY,Z\}
}
$$
These wiring diagrams are actually morphisms $\phi_1, \phi_2,\phi_3$ represented by the three cospans above.
\footnote{In $\phi_3$ the set $C$ is supposed to have two elements (i.e. $XY$ is the name of one element), and $f_3(X)=f_3(Y)=XY$.}

For example, we consider $\phi_1\taking(B_1,B_2)\to C$, where $B_1=\{X,Y,Z\}, B_2=\{Z\}$, and $C=\{X,Y\}$. The relations $R_1\in\Rel_\NN(B_1)$ and $R_2\in\Rel_\NN(B_2)$ are defined by 
$$R_1=\{(X,Y,Z)\in\NN^3\|X*Y=Z\}\hsp\tn{and}\hsp R_2=\{Z\in\NN\|Z=9\}.$$
We apply Diagram (\ref{dia:composition formula for singly typed}) to our case
$$\xymatrix{
&\NN^{\{X,Y\}}\\
\NN^{\{X,Y,Z\}}\times\NN^{\{Z\}}&\NN^{\{X,Y,Z\}}\ar[u]\ar[l]\\
R_1\times R_2\ar[u]&P\lrlimit\ar[l]\ar[u]
}
$$
where $P=\{(X,Y,Z)\in\NN^3\|X*Y=Z\;\tn{ and }\;Z=9\}$. The $\phi_1$-conjunction of $R_1$ and $R_2$ is the image of the function $P\to\NN^{\{X,Y\}}$, namely
$$\Rel_\NN(\phi)(R_1,R_2)=\{(X,Y)\in\NN^2\|X*Y=9\},$$
as desired.

All of these $(\phi_1,\phi_2$, and $\phi_3$) can be thought of as conjunctive queries in which the domain is $\NN$; we will discuss this further in Section \ref{sec:databases}.

\end{example}

\begin{remark}\label{rem:strangeness}

The $\mcS$-algebra $\Rel_A$ has somewhat counter-intuitive behavior when it comes to inputting empty relations into a wiring diagram. Given a morphism $\phi\taking X_1,\ldots,X_n\to Y$ in $\mcS$, we can apply the function $\Rel_A(\phi)$ to any $n$-tuple 
$$((R_1\ss A^{X_1}),\ldots,(R_n\ss A^{X_n}))\in\Rel_A(X_1)\times\cdots\times\Rel_A(X_n)$$
 of relations and return a relation $(S\ss A^Y)\in\Rel_A(Y)$. The strangeness is that if any of the $R_i$ are empty, then the $\phi$-conjunction will also be empty, even if there is no connection between any wire of $X_i$ and a wire of $Y$. 

There are other algebras, related to $\Rel$, having different behavior with respect to empty relations. The most intuitive or useful one may be something other than $\Rel$, but we decided for clarity to focus on $\Rel$ because its construction (see Example \ref{ex:singly typed relational alg}) is as straightforward as possible.

\end{remark}

\comment{%2013/04/26. Don't need category of elements anywhere.

Instead of working with $\Rel_A$ as an algebra, we can work with it as an operad in its own right using the following definition.

\begin{definition}

Let $\mcO$ be an operad. The {\em category of elements} functor, denoted  
$$\dispInt\taking\mcO\sets\to\Oprd/\mcO,$$
sends each algebra $R\taking\mcO\to\Sets$ to an operad $\int(R)\To{\pi_R}\mcO$ over $\mcO$ as follows. We define 
\begin{align*}
\Ob(\dispInt R)&:=\{(x,r)\|x\in\Ob(\mcO),r\in R(x)\}.\\
\Hom_{\int R}((x_1,r_1),\ldots,(x_n,r_n);(y,s))&:=\{\phi\in\mcO_n(x_1,\ldots,x_n;y)\|\phi(r_1,\ldots,r_n)=s\}.
\end{align*}
The operad functor $\pi_R$ sends $(x,r)$ to $x$ and is obvious on morphisms.

\end{definition}

}%2013/04/26. Don't need category of elements anywhere.

\begin{application}

The relational $\mcS$-algebra might be useful in machine learning applications. A goal there could be phrased as follows. Suppose that $R\ss A^Y$ is a known $|Y|$-ary relation that we are trying to investigate. An {\em explanation} of the phenomenon $R$ could be cast as a morphism $\phi\taking X_1,\cdots,X_k\to Y$ and relations $P_1\ss A^{X_1}$ through $P_k\ss A^{X_k}$ such that $\mcS(\phi)(P_1,\ldots,P_k)=R$. We can search through the space of ways to ``wire up" known relations to produce the desired relation. 

A similar idea should be prevalent as a basic constituent of any design process. The object to be designed is understood as a relationship between its various aspects, and we want to produce it using existing objects, each of which is also a relationship between its various aspects (see \cite{HK}). In other words, the outer relation must be manifest as a wiring diagram of known relations.

\end{application}

\begin{application}

The singly-typed wiring diagram operad $\mcS$ and its algebras can easily work with time-dependent phenomena. For example, suppose we take $A=C^\infty(\RR)$, the set of smooth functions from $\RR\to\RR$; we speak as though the independent variable represents time. Relations on $A$ can model such phenomena as 
\begin{itemize}
\item time delay,
\item time-wise addition, subtraction, trigonometric functions, etc., and
\item differential equations and differential relations.
\end{itemize}
For example we can represent a 4-second time delay using a relation $\tn{Delay}_4\ss A^2$ where 
$$\tn{Delay}_4=\{(f,g)\in A\times A\|\forall t\in\RR,\ \ f(t)=g(t+4)\}.$$
Time-wise addition can be given by the relation $\{(f,g,h)\in C^\infty(\RR)\|\forall t\in\RR, f(t)+g(t)=h(t)\}$. Similarly, differential equations or differential relations such as $\{(f,g)\in C^\infty(\RR)\|\forall t\in\RR, (f'(t),g'(t))\neq (0,0)\}$.

These are just standard relations. The idea here is that wiring diagrams are a nice formalism for putting them together. To consider functions of more than one variable, one would simply transfer the above ideas to the typed WD operad discussed in Section \ref{sec:typed wd operad}.

\end{application}

\section{Brief analysis of relational $\mcS$-algebras}\label{sec:brief analysis}\setcounter{subsection}{1}\setcounter{subsubsection}{0}

To study the relational $\mcS$-algebra $\Rel_A$, it would be interesting to find algebraic invariants, i.e. algebra morphisms $\Rel_A\to X$ for some $\mcS$-algebras $X$. In other words, suppose we could associate say a complexity level, e.g. a number, to each relation in a database in such a way that the complexity levels add when relations are wired together. This would provide a useful invariant by which we would know, say, the run-time of a query before running it. Unfortunately, it seems that there are very few interesting $\mcS$-algebra morphisms {\em out of $\Rel$}. Example \ref{ex:algebras on U} and Proposition \ref{prop:no alg hom Rel to N} explain that there are no interesting additive invariants (``additive" in the sense of the operation of commutative monoids) of $\Rel$.

\begin{example}\label{ex:algebras on U}

By Example \ref{ex:commutative monoids}, there is an operad $\mcE$ such that the category of commutative monoids is equivalent to the category of $\mcE$-algebras. By Lemmas \ref{lemma:pullback of algebras} and \ref{lemma:terminal operad}, we can import any $\mcE$-algebra to an $\mcS$-algebra by pulling back along the unique functor $t\taking\mcS\to\mcE$. However, the resulting $\mcS$-algebra will not sensitive to the precise wiring pattern but only to the number $n$ of inner stars in a wiring diagram $\phi\taking x_1,\ldots,x_n\to y$. Regardless, if $\CMon$ is the category of commutative monoids, we have a functor 
\begin{align}\label{dia:cmon to s-alg}
t^*\taking\CMon\to\mcS\sets.
\end{align}

What would it mean to give an algebra morphism $\Rel_A\to t^*(\tnN)$, say if $\tnN=(\NN,0,+)$? In terms of databases (see Example \ref{ex:three queries revisited} and Section \ref{sec:databases}) it would mean the following. We would have to assign to every relation a natural number in such a way that when a relation $Y$ is written as a conjunctive query on relations $X_1,\ldots,X_n$, then the number assigned to $Y$ equals the sum of the numbers assigned to the $X_i$. One way of doing this would be to assign $0$ to every relation, but this is uninteresting. Proposition \ref{prop:no alg hom Rel to N} shows that there generally are no other algebra homomorphisms $\Rel_A\to\tnN$.

\end{example}

\begin{proposition}\label{prop:no alg hom Rel to N}

Let $A$ be a set with cardinality $|A|\geq 2$, and $\Rel_A$ the $A$-typed relational $\mcS$-algebra. Let $(N,0,+)$ denote any commutative monoid and let $\mcN=t^*(\tnN,0,+)$ denote the $\mcS$-algebra that assigns the set $N$ to every singly-typed star $X\in\Ob(\mcS)$ and summation to morphisms $\phi$ in $\mcS$, as above. The only algebra morphism $p\taking\Rel_A\to\mcN$ is trivial.

\end{proposition}

\begin{proof}

We will show that $p$ sends every relation to $0\in N$ in four paragraphs corresponding to the following four steps. First we will show this for any complete relation, then for any non-empty relation, then for the empty 1-ary relation, and finally for the empty relation of any arity.

Let $Y\in\Ob(\mcS)$ be any finite set and let $K_Y\in\Rel_A(Y)$ denote the complete relation $(A^Y\ss A^Y)$. Consider the  morphism $\phi\in\Hom_\mcS(;Y)$ given by the cospan $\emptyset\to Y\From{\id_Y}Y$. Then the image under $\Rel_A(\phi)\taking\singleton\to\Rel_A(Y)$ of the unique element in the domain is $K_Y$. On the other hand, the image under $\mcN(\phi)\taking\singleton\to\mcN(Y)=N$ of the unique element in the domain is the empty sum $0\in N$. By naturality, any algebra morphism $p\taking\Rel_A\to\mcN$ must assign $p(Y)(K_Y)=0$. 

Now let $X_1,Y\in\Ob(\mcS)$ be any finite sets, and let $(R\ss A^{X_1})$ be any element of $\Rel_A(X_1)$ such that $R\neq\emptyset$. Consider the morphism $\psi\in\Hom_\mcS(X_1;Y)$ given by the cospan $X_1\to X_1\amalg Y\from Y$. The image under $\Rel_A(\psi)\taking\Rel_A(X_1)\to\Rel_A(Y)$ of $R$ is again $K_Y$, whereas the function $\mcN(\psi)\taking N\to N$ is the identity. By naturality, $p(X_1)(R)=0$.

Now let $X_1=X_2=Y=\singleton$, and let $a_1\neq a_2\in A$ be two distinct elements. Let $R_1=\{a_1\}\ss A^{X_1}$ and let $R_2=\{a_2\}\ss A^{X_2}$. Consider the morphism $\upsilon\taking (X_1,X_2)\to Y$ given by the cospan $\singleton\amalg\singleton\to\singleton\from\singleton$. Since $R_1$ and $R_2$ are nonempty, but $\Rel(\upsilon)(R_1,R_2)=\emptyset$, we have 
$$p(Y)(\emptyset)=p(X_1)(R_1)+p(X_2)(R_2)=0+0=0.$$

Finally, let $X_1$ be any finite set, $Y=\singleton$. Let $\psi\taking X_1\to Y$ be given by the cospan $X_1\to X_1\amalg Y\from Y$. If $\emptyset\ss A^{X_1}$ is the empty relation then its image under $\Rel_A(\psi)$ is $\emptyset\in\Rel_A(Y)$. As above we have $p(X_1)(\emptyset)=0$.

\end{proof}

Proposition \ref{prop:no alg hom Rel to N} shows that commutative monoids, considered as as $\mcS$-algebras, are not useful invariants of the $\mcS$-algebra $\Rel_A$. In fact, non-trivial, non-injective algebra morphisms out of $\Rel_A$ seem to be rare, and this provides a nice challenge to interested readers---see Conjecture \ref{conj:quotient-free}. Example \ref{ex:equivalence relations} provides another attempt, and Proposition \ref{prop:equivs not informative} shows that it fails to be any more informative. However, Example \ref{ex:equivalence relations} is interesting in its own right.

\begin{example}\label{ex:equivalence relations}

In this example, we discuss another $\mcS$-algebra in which local equivalence relations are wired together to form a global one. Given a finite set $X\in\Ob(\mcS)$, let $Eq(X)$ denote the set of equivalence relations on $X$. We can think of an equivalence relation $\sim$ on $X$ in terms of the surjection $X\to X/\sim$, and thus depict each equivalence relation on $X$ by connecting the appropriate wires internally. Here is a picture of an equivalence relation ($g\sim h\sim a, b\sim c, d\sim e$) on $X=\{a,b,c,d,e,f,g,h\}$:

\begin{center}
\includegraphics[height=1.5in]{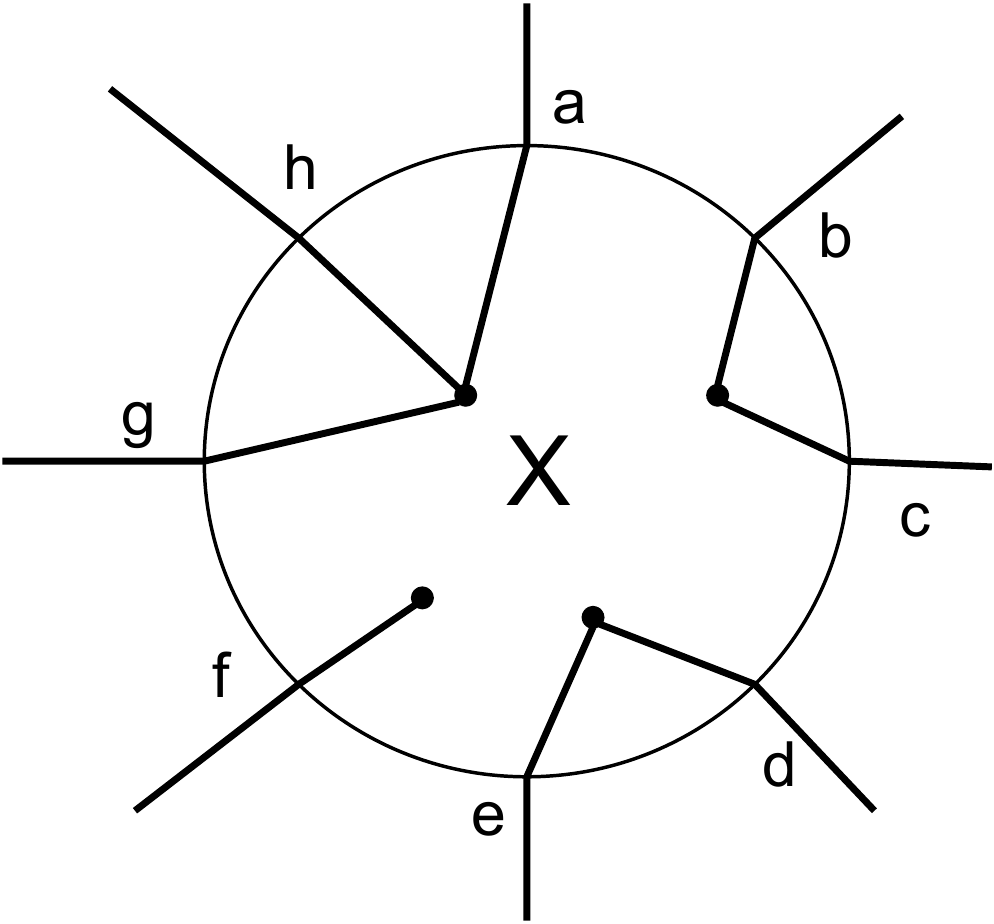}
\end{center}

Suppose $\phi\taking (X_1,\ldots,X_n)\to Y$ is a morphism in $\mcS$ given by the cospan $X_1\amalg\cdots\amalg X_n\To{f}C\From{g}Y$. It is easy to see visually how $Eq(\phi)\taking Eq(X_1)\times\cdots\times Eq(X_n)\to Eq(Y)$ should act: if two wires are connected through the wiring diagram, then we consider them to be connected globally. 
\begin{center}
\includegraphics[height=2in]{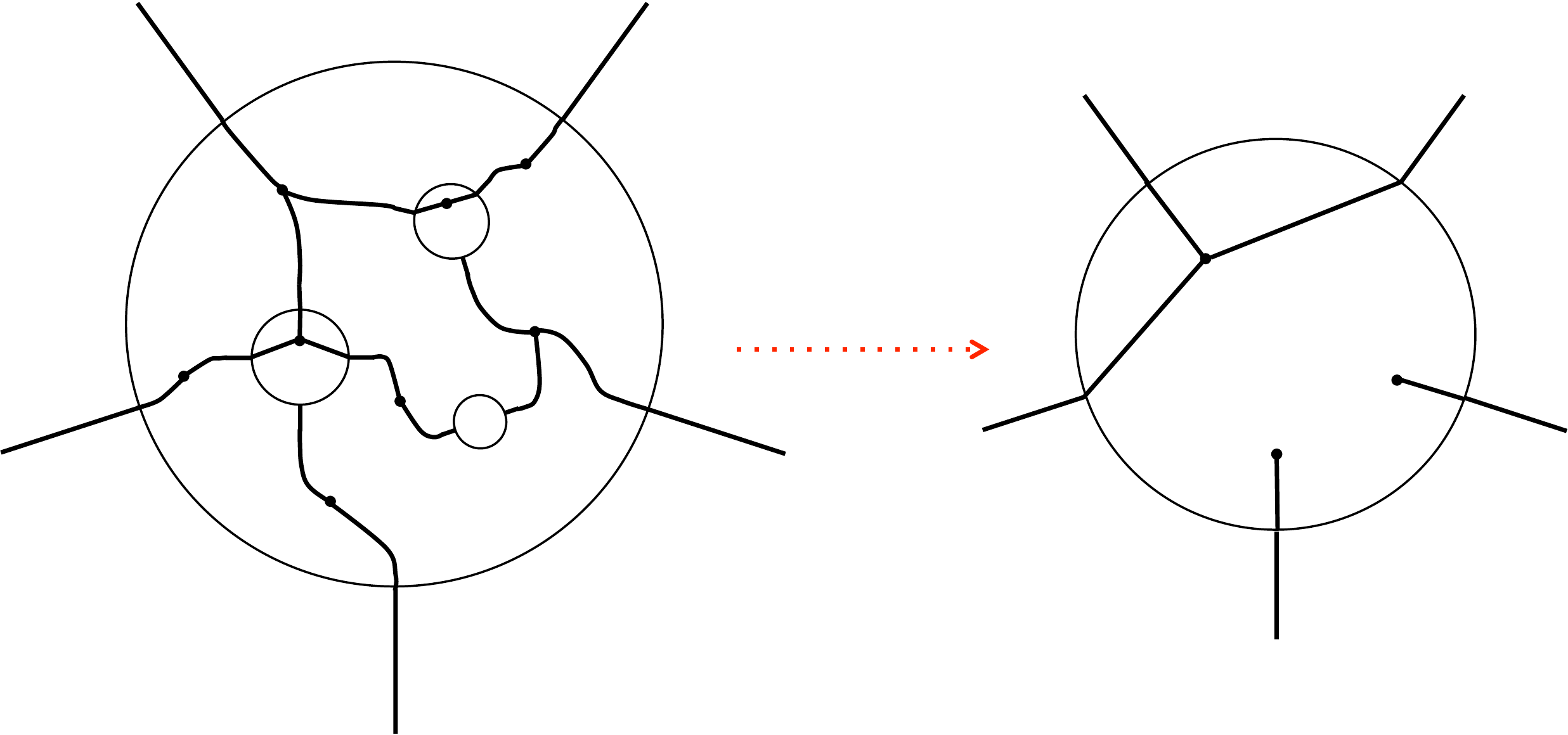}
\end{center}
Of course, this can all be done formally. Briefly, given a cospan $\coprod_iX_i\To{f} C\From{g} Y$ and surjections $X_i\surj Q_i$ for each $i\in\disc{n}$, we form the surjection $\coprod_iX_i\surj\coprod_iQ_i$, push it out along $f$ to give some new set $P$, factor the composite $Y\surj C\inj P$ as a surjection followed by an injection, and take the surjective part $Y\surj C$ as our induced equivalence relation on $Y$.

\end{example}

\begin{proposition}\label{prop:equivs not informative}

Let $A\in\Ob(\Set)$ be a set with cardinality $|A|\geq 2$. Let $Eq$ and $\Rel_A$ be the $\mcS$-algebras discussed above. The only algebra morphism $p\taking\Rel_A\to Eq$ is trivial. 

\end{proposition}

\begin{proof}

The trivial equivalence relation $\Delta_X\taking X\to X\times X$ on a set $X$ is the one corresponding to the surjection $\xymatrix@1{X\ar@{->>}[r]^{\id_X}&X}$. For any finite set $X\in\Ob(\mcS)$, consider the morphism $\phi\in\Hom_\mcS(;X)$ given by the cospan $\emptyset\to X\From{\id_X}X$. It induces an element $\Rel_A(\phi)(\ast)\in\Rel_A(X)$ and an element $Eq(\phi)(\ast)\in Eq(X)$, namely the complete relation $K_X$ and the trivial equivalence relation $\Delta_X$, respectively. It follows by naturality that under any $\mcS$-algebra morphism $p\taking\Rel_A\to Eq$, we have $p(K_X)=\Delta_X$.

Let $X=\disc{2}$, and let $R\ss A^X$ be any relation; we will show that $p(R)=\Delta_X$. For contradiction, assume that $p(R)$ is the only non-trivial equivalence relation on $\disc{2}$, namely the one given by the surjection $m\taking\disc{2}\to\disc{1}$. Let $\psi\taking(\disc{2},\disc{2})\to\disc{2}$ be the cospan $\{1_a,2_a\}\amalg\{1_b,2_b\}\to\{1_a,1_b,2_{ab}\}\from\{1_a,1_b\}$ be the wiring diagram as drawn to the left below:
\begin{center}
\includegraphics[height=1in]{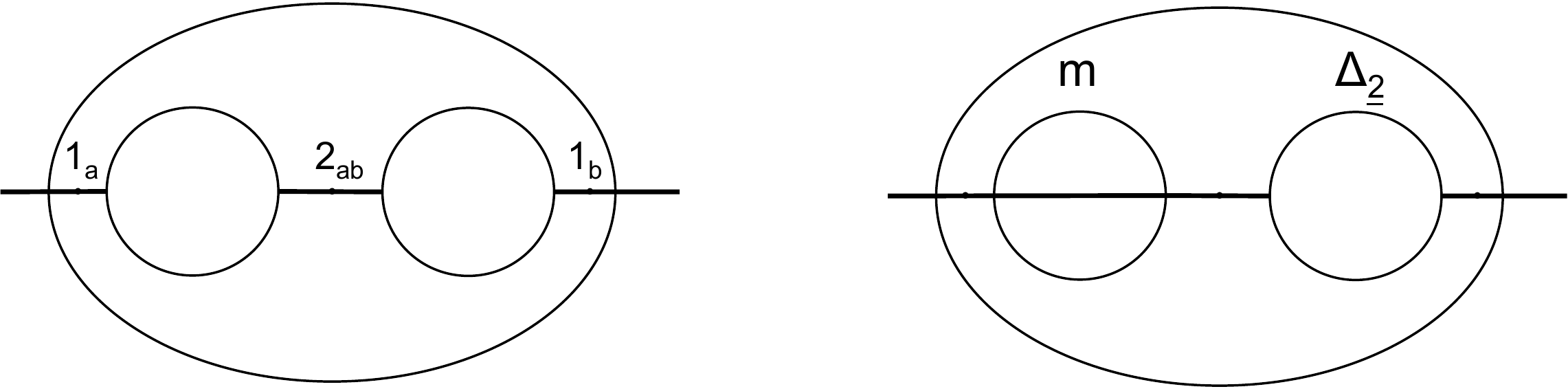}
\end{center}
There are two cases. If $R\neq \emptyset$ then $R=\Rel_A(\psi)(R,K_{\disc{2}})$ and (as seen in the right-hand diagram above) we have
$$
m=p(R)=p(\Rel_A(\psi)(R,K_{\disc{2}}))=Eq(\psi)(p(R),p(K_{\disc{2}}))=Eq(\psi)(m,\Delta_{\disc{2}})=\Delta_{\disc{2}},
$$
a contradiction. If $R=\emptyset$ we argue as in the proof of Proposition \ref{prop:no alg hom Rel to N}: let $a_1\neq a_2\in A$, let $R_1=\{(a_1,a_1)\}$ and $R_2=\{(a_2,a_2)\}$. Although both are nonempty, their $\psi$-conjunction is empty, so we have that 
$$m=p(\emptyset)=p(\Rel_A(\psi)(R_1,R_2))=Eq(\psi)(p(R_1),p(R_2))=Eq(\psi)(\Delta_{\disc{2}},\Delta_{\disc{2}})=\Delta_{\disc{2}},$$
again a contradiction.

Finally, let $X\in\Ob(\mcS)$ be any finite set, let $R'\ss A^X$ be any relation, and let  $m':=p(X)(R')\in Eq(X)$ be its image under $p$; we want to show that $m'\taking X\surj Q$ corresponds to the trivial equivalence relation $\Delta_X$, i.e. that $m'$ is injective. For contradiction, assume $m'$ is not injective, and let $x_1,x_2\in X$ be such that $m'(x_1)=m'(x_2)$, and consider the morphism $\phi\taking X\to\disc{2}$ given by the cospan $X\To{\id_X}X\From{(x_1,x_2)}\disc{2}$. Let $R=\Rel_A(\phi)(R')$ be the induced relation and let $m=Eq(\phi)(m')$ be the induced equivalence relation. Note that $m\taking\disc{2}\to\disc{1}$ is nontrivial, and that we have $p(R)=m$, contradicting the work above.

\end{proof}

It may be that there are no good invariants one can extract from the relational $\mcS$-algebra. Indeed, after considering the matter for a while we conjecture this to be the case. We make this precise using the following definition.

\begin{definition}

Let $\mcO$ be an operad and let $R\taking\mcO\to\Sets$ be an algebra. We say that $R$ is {\em quotient-free} if, for any other $\mcO$-algebra $P$ and epimorphism $f\taking R\to P$, either $f$ is an isomorphism or $P$ is isomorphic to the terminal $\mcO$-algebra $\singleton^\mcO$.

\end{definition}

\begin{conjecture}\label{conj:quotient-free}

Let $\mcS$ be the singly typed wiring diagrams operad, let $A\in\Ob(\Set)$ be a set, and let $\Rel_A\in\Ob(\mcS\sets)$ be the relational algebra of type $A$. We conjecture that $\Rel_A$ is quotient-free.

\end{conjecture}

The analysis above is all in the case $|A|\geq 2$, whereas this is not assumed in Conjecture \ref{conj:quotient-free}. But one can check easily enough that if $|A|\leq 1$ then $\Rel_A$ is quotient-free (although note the subtlety that occurs for the object $\emptyset\in\Ob(\mcS)$, e.g. by (\ref{dia:rel definition}) we have $|\Rel_\emptyset(\emptyset)|=2$).

\section{The typed wiring diagrams operad}\label{sec:typed wd operad}

The typed wiring diagrams operad is like the untyped version from Example \ref{ex:singly-typed wd}, except that each wire and each cable carries a set of possible values, and these sets must match for wires or cables that are soldered together. In Section \ref{sec:definition of typed WD} we will define $\mcT$ and compare it to $\mcS$. In Sections \ref{sec:databases} and \ref{sec:closed operads} we will discuss some applications of both operads.

\subsection{Definition of typed WD operad}\label{sec:definition of typed WD}

\begin{example}[{\bf Main Operad Variant}]\label{ex:main operad variant}

In this example we discuss the operad $\mcT$ of typed wiring diagrams. Each object of $\mcT$ consists of a finite set $X$ and a function $\tau\taking X\to\Ob(\Set)$, i.e. 
$$\Ob(\mcT)=\{(X,\tau)\|X\in\Ob(\Fin),\;\tau\taking X\to\Ob(\Set)\}\iso\Ob(\Set/\Fin).$$
We refer to each object as a {\em typed star}. If $X=\{x_1,x_2,\ldots,x_n)$, we may denote $(X,\tau)$ by 
$$\{x_1\colon\tau(x_1), \;x_2\colon\tau(x_2),\ldots, x_n\colon\tau(x_n)\}.$$ 
Each $x\in X$ is called a {\em wire of $X$}, the set $\tau(x)$ is called the {\em value type of $x$}, and each element $v\in\tau(x)$ is called a {\em value on $x$}.

We draw an object such as $\{a\colon A,\; b\colon B,\; c\colon C,\; d\colon D,\; e\colon E\}$ as follows:

\begin{align}\label{dia:labeledOb}
\parbox{2in}{\includegraphics[width=2in]{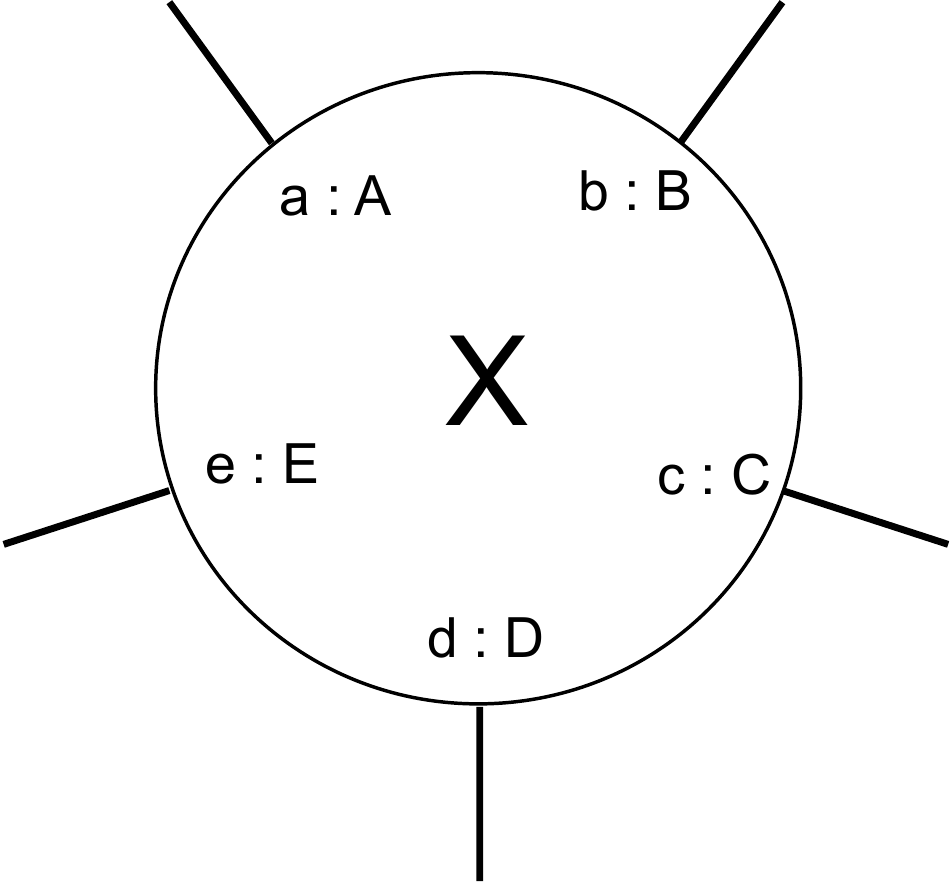}}
\end{align}
The order of the elements around the circle does not matter: $(X,\tau)$ is just a set of five sets.

Let $n\in\NN$ be a natural number, and let $(Y,\upsilon)\in\Ob(\mcT)$ and $(X_1,\tau_1),\ldots,(X_n,\tau_n)\in\Ob(\mcT)$ be typed stars. A morphism 
$$\phi\taking(X_1,\tau_1),\ldots,(X_n,\tau_n)\to(Y,\upsilon)$$
consists of a cospan over $\Ob(\Set)$, i.e. a finite set $C\in\Ob(\Fin)$, a function $\sigma\taking C\to\Ob(\Set)$, and functions $f,g$ as below, making the diagram commute: 
\begin{align}\label{dia:typed morphism}
\xymatrix{
&Y\ar[d]^g\ar@/^1pc/[ddr]^{\upsilon}\\
X_1\amalg\cdots\amalg X_n\ar@/_1pc/[rrd]_{\amalg_i\tau_i}\ar[r]_-f&C\ar[dr]^\sigma\\
&&\Ob(\Set)
}
\end{align}
We refer to such a morphism as a {\em typed wiring diagram}; we refer to the domain objects as the {\em inner stars}, the codomain object as the {\em outer star}, and the set $C$ as the {\em set of cables}. We similarly follow Example \ref{ex:singly-typed wd} in terms of both notation and depiction, drawing our typed wiring diagrams just like untyped wiring diagrams (\ref{dia:unlabeledMorphism}), except perhaps including the typing information (compare (\ref{dia:unlabeledOb}) and (\ref{dia:labeledOb})).

Given composable arrows $\phi'$ and $\phi_1,\ldots,\phi_n$, we can arrange their cospans in the form of a W-shaped diagram and form the pushout as in (\ref{dia:untyped composition}). The pushout dutifully maps to $\Ob(\Set)$, so we have the desired composition formula. As in Example \ref{ex:singly-typed wd}, the composition formula is intuitive when drawing pictures; see Diagram (\ref{dia:untyped composition drawing}).

\end{example}

There is a functor $U\taking\mcT\to\mcS$ given by forgetting the labelings on each star. Let $A\in\Ob(\Set)$ be a set. There is an operad functor $F_A\taking\mcS\to\mcT$ given by taking a singly-typed star to the typed-star with every wire having type $A$. The composite 
$$\mcS\To{F_A}\mcT\To{U}\mcS$$ is the identity on $\mcS$.

\begin{lemma}

For any set  $A\in\Ob(\Set)$, let $F_A\taking\mcS\to\mcT$ be the morphism given above and let $F_A^*\taking\mcT\sets\to\mcS\sets$ be the functor from Lemma \ref{lemma:pullback of algebras}, and let $\Rel_A\in\mcS\sets$ be the algebra given in Example \ref{ex:singly typed relational alg}. There exists a $\mcT$-algebra $\Rel\taking\mcT\to\Sets$ such that for every set $A$ we have an isomorphism of $\mcS$-algebras $$F_A^*(\Rel)\iso\Rel_A.$$

\end{lemma}

\begin{proof}

We simply follow Example \ref{ex:singly typed relational alg} in defining {\em the relational $\mcT$-algebra}, $\Rel\taking\mcT\to\Sets$. Given a typed star $\ol{X}=(X,\tau)\in\Ob(\mcT)$, let $\Pi_X(\tau)$ denote the set 
$$\Pi_X(\tau):=\prod_{x\in X}\tau(x).$$
We assign $\Rel(\ol{X})=\{R\ss\Pi_X(\tau)\}$, the set of {\em relations of type $\ol{X}$}. 

Suppose given a typed wiring diagram $\phi\taking(X_1,\tau_1),\cdots,(X_n,\tau_n)\to(Y,\upsilon)$ as in (\ref{dia:typed morphism}) and relations $R_i\in\Pi_{X_i}(\tau_i)$ for each $i\in\ul{n}$. Analogous with the singly-typed case (\ref{dia:composition formula for singly typed}) we form the pullback diagram 
$$
\xymatrix{
&\Pi_Y(\upsilon)\\
\Pi_{X_1}(\tau_1)\times\cdots\times \Pi_{X_n}(\tau_n)&\Pi_C(\sigma)\ar[l]_-{\Pi_f}\ar[u]_{\Pi_g}&S\ar@{^(->}@/_1pc/[ul]\\
R_1\times\cdots\times R_n\ar@{^(->}[u]&Q\lrlimit\ar[l]\ar[u]\ar@{->>}@/_1pc/[ur]
}
$$
and take as $S$ the image of $Q$ under the induced function $Q\to\Pi_Y(\upsilon)$. We define
$$\Rel(\phi)(R_1,\ldots,R_n):=S$$
and refer to it as the {\em $\phi$-conjunction} of $R_1,\ldots,R_n$. It is easy to see that, under this specification, $\Rel$ respects the composition formula for $\mcT$ and thus is a $\mcT$-algebra.

To see that $F_A^*(\Rel)\iso\Rel_A$ is straightforward by construction.

\end{proof}

\subsection{Databases using the typed wiring diagrams operad}\label{sec:databases}

The so-called \href{http://en.wikipedia.org/wiki/Relational_model}{{\em logical structure} of a relational database} can be defined as follows. At the highest level, the logical structure consists of a finite set $P$ of predicates, each of which will eventually correspond to (i.e. holds the place for) a table. A {\em predicate} $p$ consists of a finite set of attributes, each of which is a set, i.e. $p=(n_p,A_p)$, where $n_p\in\Ob(\Fin)$ is a finite set and $A_p\taking n_p\to\Ob(\Set)$. We sometimes denote $p$ by $(A_p(1),\ldots,A_p(n_p))$. The set $n_p$ corresponds to the set of columns of $p$'s table, and for each column $i\in n_p$, the set $A_p(i)$ is the set of values that can occur in that column. Some relational databases do not allow for different data types. In such cases there is a single universe $U\in\Ob(\Set)$ and $A_p(i)=U$ for each $i\in n_p$, 

Given such a logical structure for a relational database, an {\em instance} of it consists of a relation $R_p\ss A_p(1)\times\cdots\times A_p(n_p)$ for each predicate $p\in P$. Being in the presence of an instance means that the tables we looked forward to above have been filled in with data of the appropriate types.

A {\em query} is roughly a question asked of the database. Given a set $P$ of predicates as above, a {\em conjunctive query} is often written in the following form
\begin{tabbing}
\hsp\=SELECT \hspace{.1in}\=$x_1,x_2,\ldots,x_k$\\
\>FROM \>$p_1,p_2,\ldots,p_\ell$\\
\>WHERE \>$x_{i_1}=x_{i_1}'$ AND $x_{i_2}=x_{i_2}'$ AND $\cdots$ AND $x_{i_m}=x_{i_m}'$
\end{tabbing}
Here each $x_i$ (and $x_{i_j}$ and $x_{i_j}'$) refers to a specific attribute in a specific predicate $p_1,\ldots, p_\ell$. A conjunctive query can also be thought of as follows. Start with a set of relations and take their Cartesian product (using the FROM-clause), then choose those rows for which the values in specified pairs of columns agree (using the WHERE-clause), and then project down to a subset of the columns (using the SELECT-clause). 

Wiring diagrams (i.e. morphisms in $\mcT$) give us a way to picture conjunctive queries. For example, a predicate $p=(n_p,A_p)$ is precisely the same thing as a typed star, i.e. an object $p\in\Ob(\mcT)$. 
\footnote{
In the case of a single universe $U$, we can make due with the singly-typed wiring diagrams operad $\mcS$ and the algebra $\Rel_U$.
}
A morphism $\phi\taking ((n_1,A_1),\ldots,(n_\ell,A_\ell))\to (k,x)$ consists of a cospan over $\Set$, i.e. a diagram of the form 
\begin{align*}
\xymatrix{
&k\ar[d]^g\ar@/^1pc/[ddr]^{x}\\
n_1\amalg\cdots\amalg n_\ell\ar@/_1pc/[rrd]_{\amalg_iA_i}\ar[r]_-f&C\ar[dr]^A\\
&&\Ob(\Set)
}
\end{align*}

Conjunctive queries are a bit more limited than what can be expressed using wiring diagrams.
\!\!\comment{
In particular, the morphism $\phi$ in a conjunctive query is assumed to be {\em frugal} in the sense of the following definition.

\begin{definition}\label{def:frugal}

Let $\phi\taking(X_1,\ldots,X_n)\to Y$ be a morphism in $\mcT$ represented by the cospan $(X_1\amalg\cdots\amalg X_n)\To{f}C\From{g}Y$. Let $p\taking \amalg_{i\in n}X_i\to n$ denote the function that takes each element of the disjoint union to the component in which it occurs. Let $P$ be the pushout in the following diagram:
$$
\xymatrix{
\amalg_{i\in n}X_i\ar[r]^f\ar[d]_p&C\ar[d]\\
n\ar[r]&P\lrlimit
}
$$
Then we say that the wiring diagram $\phi$ is {\em frugal} if and only if $|P|=1$.

\end{definition}

The idea of Definition \ref{def:frugal} is the following. In a wiring diagram $\phi$, the wires connect stars and cables, forming a bipartite graph. We say that $\phi$ is frugal if this graph has only one connected component. 
}
This will be explored a bit in Example \ref{ex:non-conjunctive}.

\begin{example}\label{ex:wiki database query}

Here is an example of a \href{http://en.wikipedia.org/wiki/Conjunctive_query}{\text conjunctive query} from wikipedia, written in the database query language SQL. It uses constants (`male' and `female'), whereas our above description of conjunctive queries does not, but we can still make good sense of it with wiring diagrams.
\begin{tabbing}
\hsp\=SELECT\hspace{.1in}\=L.student, L.address\\
\>FROM\>attends a1, gender g1, attends a2, gender g2, lives L\\
\>WHERE\>a1.student=g1.student AND a2.student=g2.student \\
\>AND\>L.student=g1.student AND a1.course=a2.course\\
\>AND\>g1.gender=`male' AND g2.gender=`female'
\end{tabbing}

Below in (\ref{dia:wikiQueryWD}) we show this query in the form of a wiring diagram. The five black inner stars correspond to the five tables in the FROM clause; the two red inner stars are the constants; the two wires of the outer star correspond to the two attributes in the SELECT clause; and the precise wiring pattern corresponds to the conditions in the WHERE clause.
\begin{align}\label{dia:wikiQueryWD}
\includegraphics[width=\textwidth]{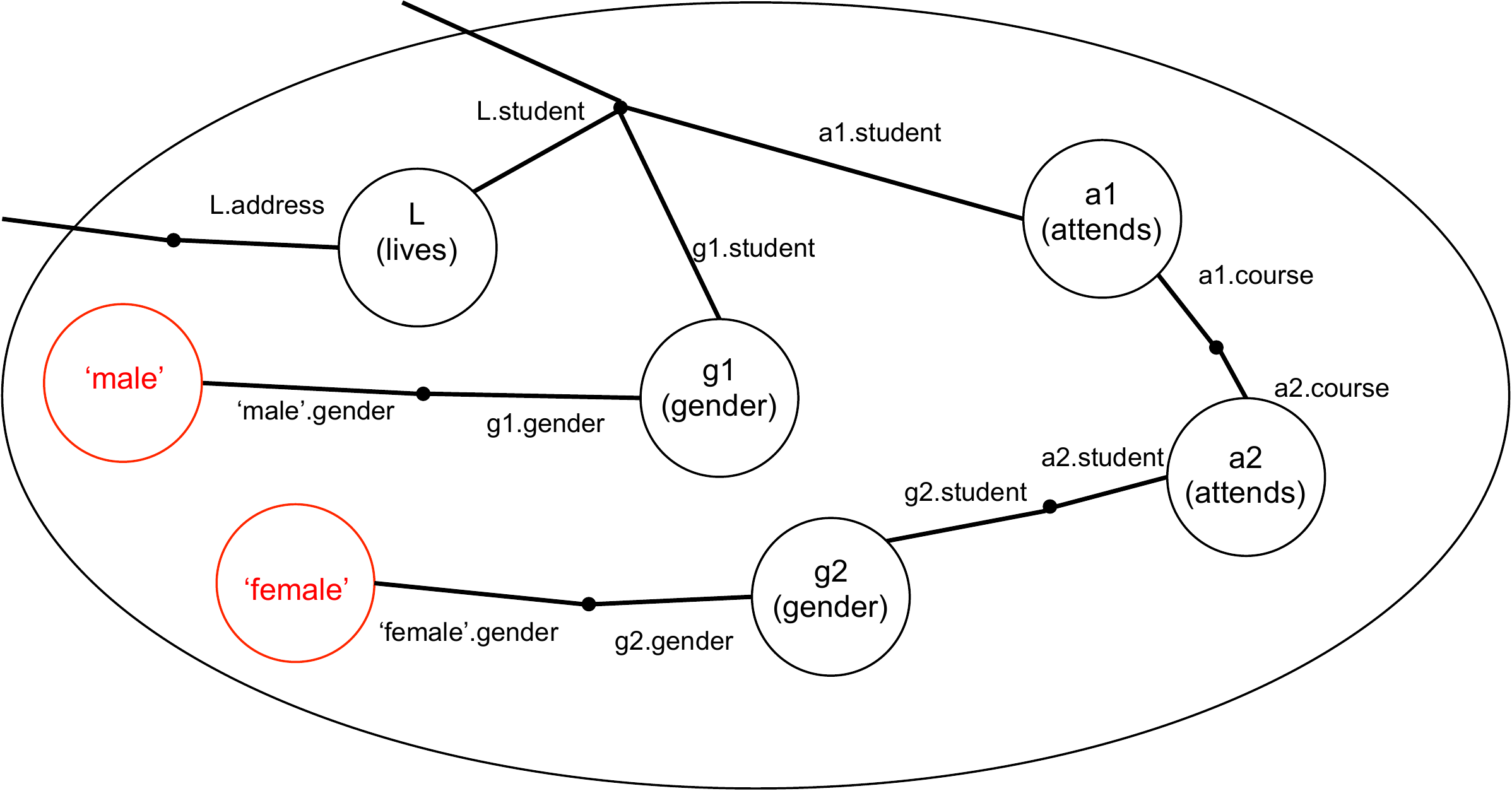}
\end{align}
The constants `male' and `female' are each simply considered as a 1-ary relation on the set of genders.

\end{example}

\begin{example}\label{ex:non-conjunctive}

Here we show a non-conjunctive query that can be made using a wiring diagram, namely $y.\exists x. R(x)$ for some unary relation $R$. The wiring diagram $\phi$ is represented by the cospan is $\{x\}\To{f}\{x,y\}\From{g}\{y\}$,
\begin{center}
\includegraphics[height=1in]{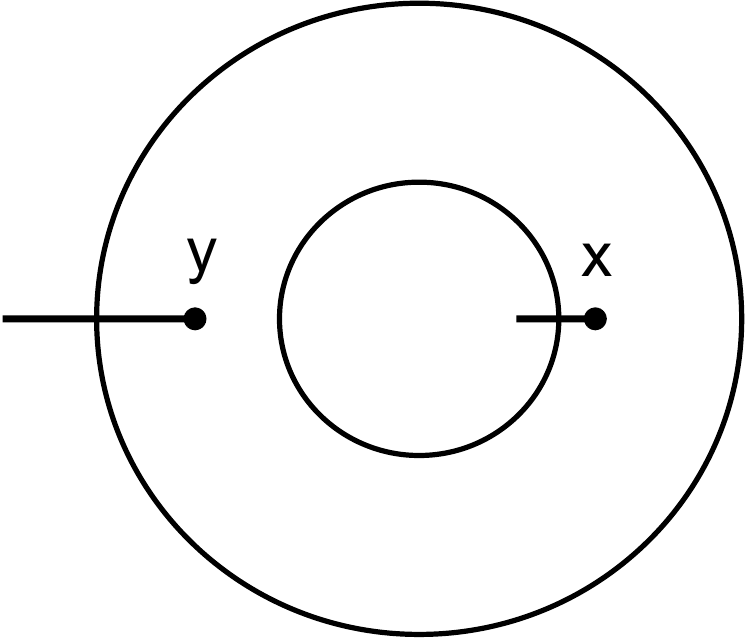}
\end{center}
and the reason it is considered not to be conjunctive is that $y$ is not in the image of $f$.

Given a relation $R\in\Rel(\{x\})$, the function $\Rel(\phi)\taking\Rel(\{x\})\to\Rel(\{y\})$ returns $(\emptyset\ss S)$ if $R=\emptyset$ and returns $(S\ss S)$ if $R\neq\emptyset$. (See Remark \ref{rem:strangeness}). That is $\Rel(\phi)(R)$ is $y . \exists x. R(x)$.

\end{example}

\subsection{Disjunctive queries}

We can add disjunctive queries to our operadic query language by considering $\Rel$ as an $\mcT$-algebra object in $\JLats$, the operad of join semi-lattices, rather than in $\Sets$. See Example \ref{ex:operads s}. To make this precise, we have the following proposition.

\begin{proposition}\label{prop:algebras in JLat}

Let $\mcT$ be the typed wiring diagrams operad, and let $\Rel\taking\mcT\to\Sets$ be the relational $\mcT$-algebra. Let $U\taking\JLats\to\Sets$ be the operad functor that sends a join-semilattice to its underlying set. Then there exists a morphism $\JRel\taking\mcT\to\JLats$ such that the following diagram of operads commutes:
$$
\xymatrix@=15pt{\mcT\ar@{-->}[rr]^{\JRel}\ar[rdd]_{\Rel}&&\JLats\ar[ldd]^U\\\\&\Sets.}
$$

\end{proposition}

\begin{proof}

Given a typed star $\ol{X}:=(X,\tau)\in\Ob(\mcT)$, where $X\in\Ob(\Fin)$ and $\tau\taking X\to\Ob(\Sets)$, we define the join-semilattice 
$$\JRel(X,\tau):=\{R\ss \prod_{x\in X}\tau(x)\},\hsp \tn{with}\hsp \bigvee_{i\in I}R_i:=\bigcup_{i\in I}R_i.$$
The underlying set $U\circ\JRel(X,\tau)$ is $\Rel(X,\tau)$, so on objects, the triangle commutes. We need to say what $\JRel$ does on morphisms $\phi\taking\ol{X}_1,\ldots,\ol{X}_n\to\ol{Y}$ in $\mcT$.

Given a cospan $(X,\tau)\To{f}(C,\sigma)\From{g}(Y,\upsilon)$, as in Diagram (\ref{dia:typed morphism}), we have a span of sets 
$$\prod_X\tau\From{\Pi_f}\prod_C\sigma\To{\Pi_g}\prod_Y\upsilon.$$
The prescription for $\Rel(\phi)$ says that for any relation $R\ss\prod_X\tau$, we pull it back to a relation on $\prod_C\sigma$ and then take its image in $\prod_Y\upsilon$. We need only to see that this process preserves the join, i.e. union of subsets. But both pullback and image are left adjoints; since unions are colimits, they are preserved. 

\end{proof}

Using algebra objects in $\JLat$, as in Proposition \ref{prop:algebras in JLat}, we can take the union of any two relations if they have the same type $(X,\tau)$. Doing so will appropriately commute with $\phi$-conjunction.

\section{Closed operads for plug-and-play and recursion}\label{sec:closed operads}

\subsection{Plug-and-play}

\href{http://en.wikipedia.org/wiki/Hot_plugging}{\text Hot-plugging}, or a variant called {\em plug-and-play}, is a computer engineering term, which refers to the capacity whereby one or more components of a computer system can be replaced without interruption to the system. We use the term only slightly differently here, in that we are not considering computers per se but any system that can be described solely by a wiring diagram. As such, rebooting is not the issue. Instead the issue is to build systems that have slots in which components can be filled in later. 

To make sense of this idea, we will use the notion of closed operads. The term should be understood in relation to Cartesian closed categories (CCCs), like those used as a denotational semantics for computer programs. The similarity between CCCs and closed operads is found in the ability to {\em curry} a morphism $f$ with multiple input types, i.e. to translate $f$ into an equivalent function that has fewer inputs but whose output is a function on the remaining inputs. The difference between closed operads and CCCs is that closed operads literally have multiple inputs, whereas CCCs in some sense {\em simulate} having multiple inputs, using Cartesian product. Since $\Set$ is a Cartesian closed category, the relationship between operads and CCCs is captured using algebras on an operad $\mcO$, i.e. operad functors $\mcO\to\Sets$ by which the multiple input types of an $\mcO$-morphism are translated into the respective Cartesian product.

The following definitions and terminology follows \cite{Man}.

\begin{definition}\label{def:closed operad}

A {\em closed operad}, denoted $(\mcO, \iHom, ev)$, consists of three constituents (A. operad, B. internal Hom-objects, C. evaluation map) conforming to one law (1. externalization). The required constituents are as follows:
\begin{enumerate}[\hsp A.]
\item There is a operad $\mcO$.
\item For each object $z\in\Ob(\mcO)$, finite set $n\in\Ob(\Fin)$, and $n$-indexed set of objects $y\taking n\to\Ob(\mcO)$, there is an object $[y\iHom z]\in\Ob(\mcO)$ called {\em the internal Hom-object associated to $y$ and $z$}.
\item For each object $z\in\Ob(\mcO)$, finite set $n\in\Ob(\Fin)$, and $n$-indexed set of objects $y\taking n\to\Ob(\mcO)$, there is an $(n+1)$-ary evaluation morphism 
$$ev_{y;z}\taking([y\iHom z],y)\too z$$
\end{enumerate}
We may denote the object $[y\iHom z]$ by $[(y_1,\ldots,y_n)\iHom z]$, when convenient or by $[y\iHom z]_\mcO$ if necessary. The law that governs these constituents is as follows:
\begin{enumerate}[\hsp 1.]
\item Suppose given an object $z\in\Ob(\mcO)$, finite sets $m,n\in\Ob(\Fin)$, and so-indexed sets of objects $x\taking m\to\Ob(\mcO)$ and $y\taking n\to\Ob(\mcO)$. The function 
$$extl_{x,y;z}\taking\mcO_{m}(x;[y\iHom z])\too\mcO_{m+n}(x,y;z),$$
defined by sending a morphism $\phi\taking x\to[y\iHom z]$ to the composite
$$x_1,\ldots,x_m,y_1,\ldots,y_n\To{\phi,\id_{y_1},\ldots,\id_{y_n}}[(y_1,\ldots,y_n)\iHom z],y_1,\ldots,y_n\To{ev_{y;z}}z,$$ 
is bijective.
\end{enumerate}

The isomorphism $extl_{x,y;z}$ defined above is called the {\em externalization function}; we denote its inverse by $intl_{x,y;z}$ and call it the {\em internalization function}.

A {\em functor} between closed operads is just a functor between their underlying operads, and an {\em algebra} on a closed operad is just an algebra on the underlying operad.

\end{definition}

\begin{remark}

The categorically-minded reader may notice that Definition \ref{def:closed operad} does not explicitly require that the internal Hom-objects $[X\iHom Y]$ be functorial in $X$ or $Y$. However, this functoriality follows from their tight relation to the actual hom-sets, $\mcO(X;Y)$, which are functorial in $X$ and $Y$ by definition. 
%In other words, given an object $Z\in\Ob(\mcO)$, objects $Y\taking\disc{n}\to\Ob(\mcO)$,  for each $i\in\disc{n}$ objects $X_i\taking\disc{m_i}\to\Ob(\mcO)$, and morphisms $\phi_i\taking X_i\to Y(i)$, we obtain a morphism $[Y\iHom Z]\to [X\iHom Z]$ as the image of the distinguished element under the following function:
%\begin{align*}
%ev_{Y;Z}\in\mcO_{n+1}([Y\iHom Z],Y;Z)&\Too{\phi}\mcO_{1+\Sigma_{i}m_i}([Y\iHom Z],X;Z)\\
%&\Too{\iso}\mcO_1([Y\iHom Z];[X\iHom Z])
%\end{align*}

\end{remark}

\begin{example}

The operad $\Sets$ of sets is closed. Indeed, given sets $Y_1,\ldots,Y_n,Z$ we can define the internal-Hom object $[(Y_1,\ldots,Y_n)\iHom Z]\in\Ob(\Sets)$ to be the set $\Hom_\Sets(Y_1,\ldots,Y_n;Z)$. 

\end{example}

\begin{proposition}

The singly-typed wiring-diagram operad $\mcS$ (respectively, the typed wiring-diagram operad $\mcT$) is closed. 

\end{proposition}

\begin{proof}

We prove this for $\mcS$ because it is less cluttered, but the proof for $\mcT$ is perfectly analogous, adding a map from everything in sight to $\Ob(\Set)$, as in Example \ref{ex:main operad variant}.

Suppose given finite sets $Y_1,\ldots,Y_n,Z\in\Ob(\mcS)$. We define the associated internal Hom-object using the coproduct, 
$$[(Y_1,\ldots,Y_n)\iHom Z]:=Y_1\amalg\cdots\amalg Y_n\amalg Z.$$ 
\begin{figure}[H]
\includegraphics[height=2.5in]{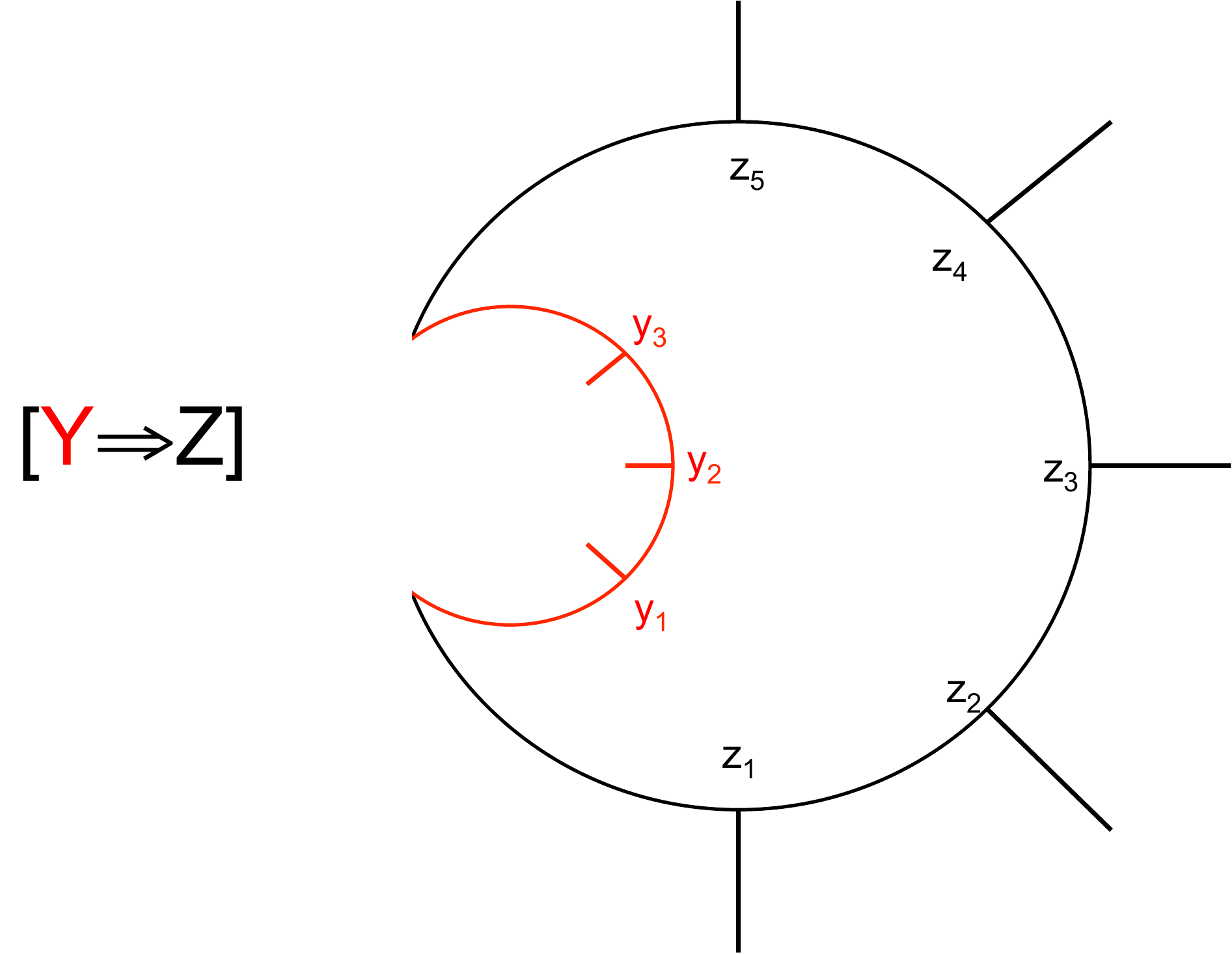}
\caption{
A picture of $[Y_1\iHom Z]$, where $n=1$, $Y_1=\{y_1,y_2,y_3\}$, and $Z=\{z_1,z_2,z_3,z_4,z_5\}$. The shape of the picture is meaningless; it is merely a suggestive depiction of the set $[Y_1\iHom Z]=Y_1\amalg Z$. 
}
\end{figure}
The evaluation morphism $ev_{Y;Z}\in\mcS_{n+1}([Y\iHom Z],Y;Z)$ is the following cospan
$$\left(\coprod_{i=1}^nY_i\right)\amalg Z\amalg\left(\coprod_{i=1}^n Y_i\right)\too\left(\coprod_{i=1}^n Y_i\right)\amalg Z\fromm Z,$$
where the first function is given by the fold map. Here is a picture of the evaluation map in the above case:
\begin{center}
\includegraphics[height=2.5in]{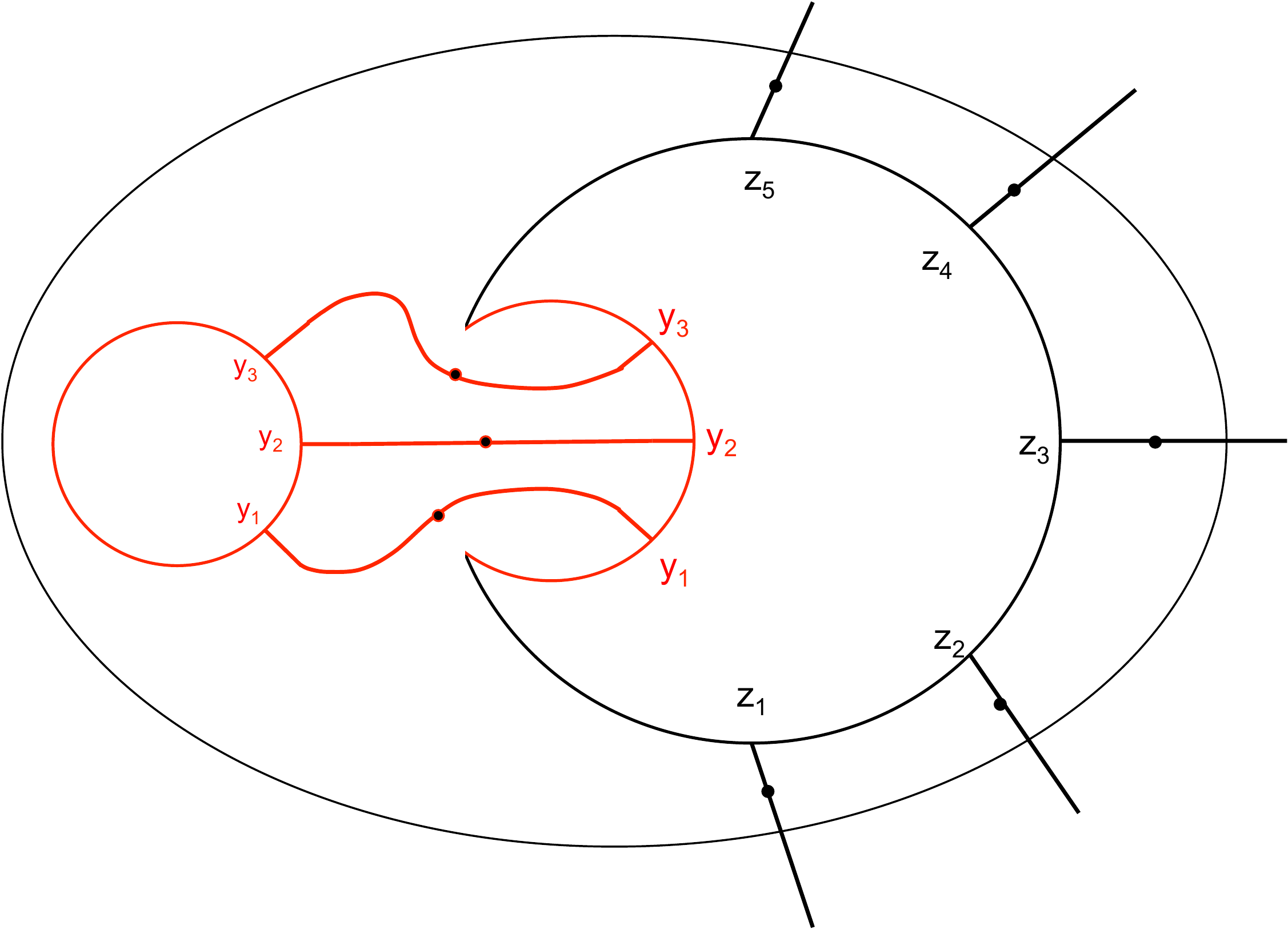}
\end{center}

To see that $(\mcS,\iHom,ev)$ satisfies the externalization condition, we begin with finite sets $X_1,\ldots,X_m,Y_1,\ldots,Y_n,Z\in\Ob(\mcS)$ and a morphism $\phi\taking X\to [Y\iHom Z]$ given by the cospan $X\To{f}C\From{g} [Y\iHom Z]$. We must show that the externalization function $extl_{X,Y;Z}\taking\mcS_m(X;[Y\iHom Z])\to\mcS_{m+n}(X,Y;Z)$ is a bijection, in accordance with Definition \ref{def:closed operad}. By definition, the morphism $extl_{X,Y;Z}(\phi)\taking (X,Y)\to Z$ is given by forming the following diagram
$$
\xymatrix{
&&Z\ar[d]\\
&[Y\iHom Z]\amalg Y\ar[r]^-{fold_Y}\ar[d]_-{g,\id_Y}&[Y\iHom Z]\ar[d]\\
X\amalg Y\ar[r]^{f,\id_Y}&C\amalg Y\ar[r]&C\lrlimit
}
$$
The bottom-right object is, according to the composition law for $\mcS$, a pushout; it is easy to check set-theoretically that this pushout is isomorphic to $C$, as written. The morphism $extl_{X,Y;Z}(\phi)$ is represented by the cospan $X\amalg Y\to C\from Z$ obtained by composites. However an easier way to understand what is going on is simply that $extl_{X,Y;Z}$ converts diagrams like the one on the left below to diagrams like the one on the right:
$$
\parbox{.4in}{\xymatrix{&Y\amalg Z\ar[d]^{g_1,g_2}\\X\ar[r]_{f}&C}}
\hsp\xymatrix{~\ar@{~>}[rr]^{extl_{X,Y;Z}}&&~}\hsp
\parbox{.4in}{\xymatrix{&Z\ar[d]^{g_2}\\X\amalg Y\ar[r]_{f,g_1}&C}}
$$
It is clear that this correspondence ($extl_{X,Y;Z}$) is a bijection, so $\mcS$ is indeed closed.

\end{proof}

\begin{remark}

We will show that in $\mcS$ (with an eye toward its algebra $\Rel$), one can speak of a morphism to the internal Hom-object $[(Y_1,\ldots,Y_n)\iHom Z]$ as an {\em explanation of $Z$, specifically in terms of $Y_1,\ldots,Y_n$}. If this is to be the case then when $n=0$ we must have that morphisms to $Z\iso[()\iHom Z]$ are explanations of $Z$ (in no specific terms). This makes sense: a morphism $(X_1,\ldots,X_n)\to Z$ is a wiring diagram, which explains the behavior of the outer relation in terms of, i.e. as a query on, the inner relations. Similarly, a morphism $X\to [Y\iHom Z]$ is the same as a morphism $(X,Y)\to Z$, so it is an explanation of $Z$ that in particular involves the star $Y$.

\end{remark}

\begin{definition}\label{def:closing transformation}

Suppose that $F\taking\mcO\to\mcO'$ is a morphism of closed operads. For any object $Y\in\Ob(\mcO)$, finite set $n\in\Ob(\Fin)$, and $n$-indexed set of objects $X\taking n\to\Ob(\mcO)$, let $ev_{X;Y}$ be the evaluation morphism in $\mcO$. Applying $F$ we obtain an $\mcO'$-morphism $F(ev_{X;Y})\taking (F[X\iHom Y]_\mcO,FX)\to FY$. Define the {\em closing transformation} of $F$ to be the $\mcO'$-morphism
\begin{align}\label{dia:closing transformation}
F_{X;Y}\taking F[X\iHom Y]_{\mcO}\too[FX\iHom FY]_{\mcO'},
\end{align}
obtained from $F(ev_{X;Y})$ by applying the internalization function.

\end{definition}

\subsection{Recursion}

We can think of recursion as what happens when a relation (e.g. a function) is ``plugged into itself". The story will go as follows. We use the fact that $\mcS$ is a closed operad and define a {\em recursive star} to be an object of the form $[Z\iHom Z]\in\Ob(\mcS)$. A {\em recursive setup} will be any wiring diagram $\phi\taking(X_1,\ldots,X_n)\to[Z\iHom Z]$ with output $[Z\iHom Z]$, and with the input stars filled in with relations $R_1\in\Rel(X_1),\ldots,R_n\in\Rel(X_n)$. Finally, a {\em recursive relation} for this setup will be any (non-empty) fixed point of the resulting function 
$$\phi(R_1,\ldots,R_n)\taking\Rel(Z)\to\Rel(Z).$$ The following picture of a recursive setup, for which the factorial function is a fixed point, may be helpful. 
\begin{align}\label{dia:factorial}
\includegraphics[width=\textwidth]{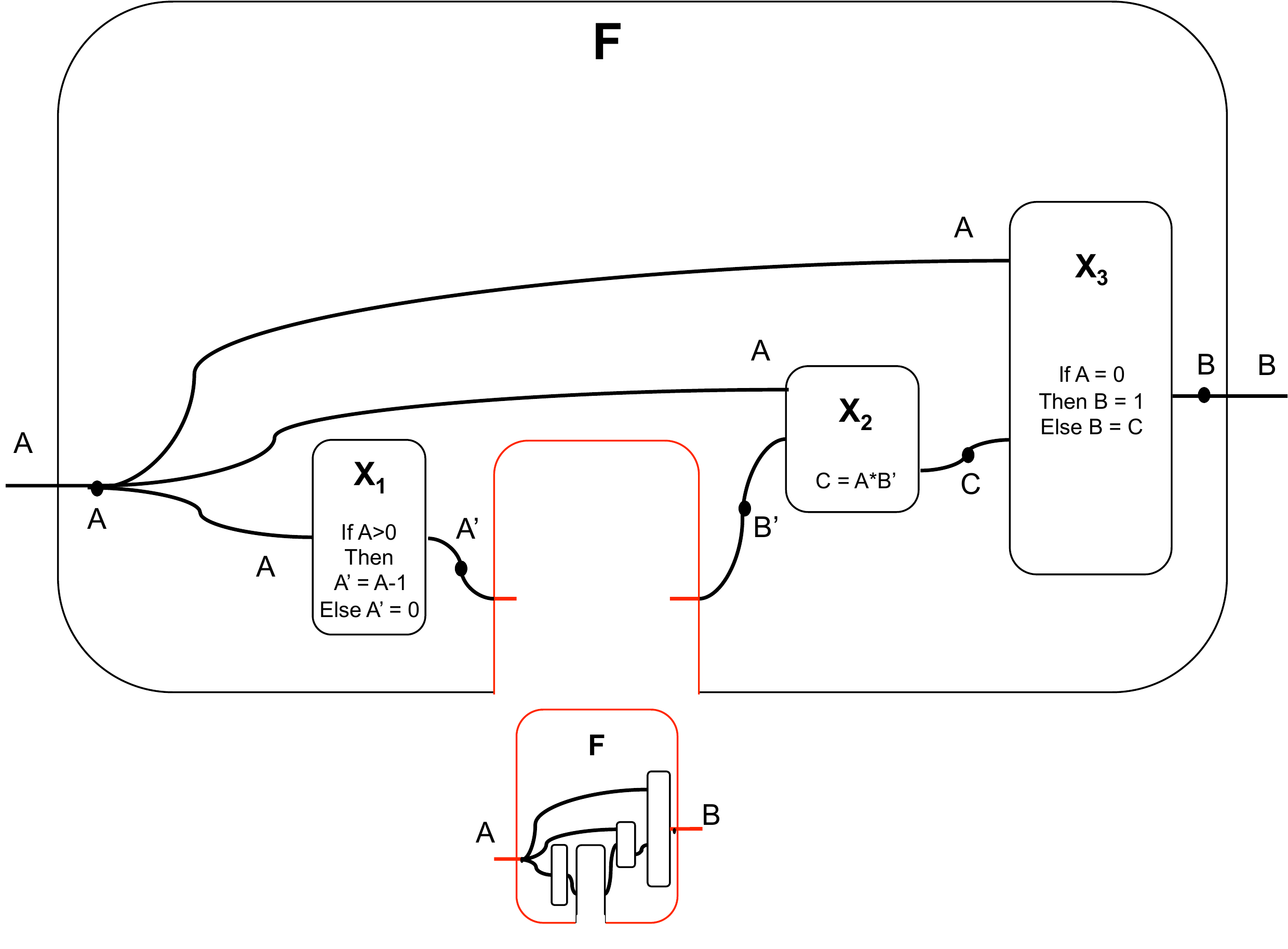}
\end{align}
We will now describe this picture in an example.

\begin{example}

First let us run through the recursive procedure being depicted in (\ref{dia:factorial}). We are describing a relation, which is supposedly the factorial function, in terms of some other relations and itself. The idea is that given a nonzero natural number $A$, we subtract 1, take the factorial $(A-1)!$, multiply this number by $A$, and output the result. If $A=0$, output 1.

The diagram shown above consists of two parts: a morphism $\phi\taking(X_1,X_2,X_3)\to[F\iHom F]$ in $\mcS$ and some chosen relations. Here $X_1:=\{A,A'\}, X_2:=\{A,B',C\}, X_3:=\{A,C,B\}$, and $F:=\{A,B\}$. The recursive star is $[F\iHom F]\iso\{A,A',B,B'\}$. The morphism $\phi$ is represented by the cospan
$$\phi\taking\{A,A'\}\amalg\{A,B',C\}\amalg\{A,C,B\}\too\{A,A',B,B',C\}\fromm\{A,A',B,B'\}.$$
The chosen relations are written in the boxes. They are: 
\begin{description}
\item [Decrement]`` If $A>0$ then $A'=A-1$ else $A'=0$ " as an element $R_1\in\Rel_\NN(X_1)$; 
\item [Multiplication] `` $C=A*B'$ '' as an element $R_2\in\Rel_\NN(X_2)$; and 
\item [Conditional] `` If $A=0$ then $B=1$ else $B=C$ " as an element $R_3\in\Rel_\NN(X_3)$.
\end{description}
Together these constitute an element $(R_1,R_2,R_3)$ of the left-hand side:
$$\Rel_\NN(X_1)\times\Rel_\NN(X_2)\times\Rel_\NN(X_3)\Too{\phi}\Rel_\NN([F\iHom F])
$$
This completes the recursive setup. From it we can form the $\phi$-conjunction, 
$$\Rel_\NN(\phi)(R_1,R_2,R_3)\in\Rel_\NN([F\iHom F]),$$ and applying to it the closing transformation (\ref{dia:closing transformation}), we have a function 
$$q\taking\Rel_\NN(F)\to\Rel_\NN(F).$$
In terms of the diagram above, $q$ is the function that takes a relation on $\NN\times\NN$, plugs it into the empty slot, and returns a relation on $\NN\times\NN$. A recursive relation for this setup is a non-empty fixed point of $q$. It turns out there is only one: the factorial relation, $B=A!$. 

\end{example}

\section{Future work}\label{sec:future work}

There are three main directions for taking this project further. The first is to investigate the mathematical properties of the wiring diagrams operad $\mcS$ (or $\mcT$) and its algebras; this is discussed in Section \ref{sec:future WD math}. The second is to use the wiring diagrams operad and its algebras to study real-world phenomena; this is discussed in Section \ref{sec:future WD phenomena}. The third is to change and expand the model, i.e. the operad $\mcS$ itself; this is discussed in Section \ref{sec:future expand WD}.

\subsection{Investigating mathematical properties of the WD operad and its algebras}\label{sec:future WD math}

It seems difficult to find any simplifying invariants of relations that are natural with respect to wiring diagrams. This was the content of Conjecture \ref{conj:quotient-free}. It would be quite interesting to know definitively whether such an invariant can exist. It would also be good to know whether relational $\mcS$-algebras have any privileged place in $\mcS\sets$, e.g. whether the algebras $\Rel_A$ can be characterized category-theoretically. Many such mathematical questions remain open about relational $\mcS$-algebras. 

It is also possible that we could greatly decrease the number of morphisms in each hom-set of $\mcS$ without having much affect on the category of algebras. For example, there is a sub-operad $\mcS'\ss\mcS$ in which each wiring diagram is required to have exactly one connected component, an upshot of which is that all hom-sets in $\mcS'$ are finite. We have not yet investigated how the category $\mcS'\sets$ compares to $\mcS\sets$.

The wiring diagrams operad $\mcS$ itself has more structure than we have considered here. First of all, the objects of $\mcS$ are finite sets, but we never discuss how functions between them enter the picture. Second, the morphisms $\phi\taking X\to Y$ in $\mcS$ are cospans in $\Fin$ of the form $X\to C\from Y$, but we never discuss the fact that there is a whole category of such morphisms (i.e. we can discuss morphisms of cospans). In other words, $\mcS$ can perhaps be lifted to something like an operad object in $\Cat$. 

It would also be useful to investigate whether aggregation, such as taking sums or averages in a database, can be understood in terms of the wiring diagrams operad $\mcS$ or $\mcT$, perhaps using the fact that they are closed. That is, we have a way to speak about functions from relations to relations, and aggregation may be phrased in this way. 

In that vein, it would be valuable to formalize more of relational database theory, including natural keys, normal forms, negation, nulls, queries etc., in the language of operads. 

\subsection{Using the WD operad and various algebras to study real-world phenomena}\label{sec:future WD phenomena}

One of the most important questions, in the author's opinion, is how structure and function are related in the human brain. If the brain is understood as a collection of neurons wired together at synapses, then we are tasked with explaining how complex functionality like walking and calculating can be housed in such a structure. The wiring diagrams operad can serve as a language for discussing the structure of the brain, in the sense that we need not think of neurons as indivisible objects, nor do we need to step from the neuron level directly to the brain level. The brain can be divided into different agencies, such as the visual cortex and the auditory cortex, and each of these can be broken down further. The WD operad allows us to zoom into any given part (either topological region or functional subset) and then connect our findings about it with those about other regions.

Our findings about each region, be it neuron, whole brain, or anything in between, might be phrased in terms of algebras over the operad $\mcS$. In this paper we used the relational $\mcS$-algebra, which would be relevant if we think of each portion of the brain as managing a relationship between its inputs and outputs. But in general, choosing an algebra allows us to do two things. First, we get to assign to each agency of the brain a set of possible states, and second, we get to assign a rule by which a state on any given whole emerges from the states on its constituent parts.

Besides using the WD operad to discuss the brain, we may also find it useful in understanding the process of design. It seems that intelligent animals can form a concept of what they want to see emerge in the world, and then work to create it. This requires a constant back-and-forth between what is known and what is intended (see \cite{HK}). We build a whole new object with certain properties by putting together known objects in creative patterns. The wiring diagrams operad or something related could be useful to frame the discussion, because it carefully delineates between the wiring pattern (operad), and how large-scale function emerges from small-scale function via that wiring pattern (algebra). 

\subsection{Expanding the model}\label{sec:future expand WD}

The wiring diagrams operad in this paper provides a simple picture for how smaller entities can be put together to form larger ones, but it is certainly not the only operad that would fit such a description. For example, a directed version of the wiring diagrams operad with an algebra of functions has been shown to exist and a paper on the subject \cite{SS} is in preparation. Directed wiring diagrams may be useful in describing processes, or recipes, whereby one process (taking inputs to outputs) is made by an ordered sequence of steps (each of which involving its own inputs and outputs).

We may also want to expand the wiring diagrams operad so that the actual spacing between inner stars, or the length of the wires, becomes part of the equation, so to speak. For example, the classical ``little cubes operad" would become a suboperad of such a {\em topological wiring diagrams} operad. The length of wires could encode Shannon-style information theory concepts, giving rise to a richer theory. 

In a totally different direction, one could study the mechanical properties of topological wiring diagrams: if we know how each inner star deforms under stress, and we know how the outer star is composed of inner stars, then perhaps we can calculate how the outer star will deform under stress. This framework could prove useful to materials scientists.

\bibliographystyle{amsalpha}

\end{document}